%% file: hybrid-adeq.tex
\pdfoutput=1
\PassOptionsToPackage{colorlinks,linkcolor={blue},citecolor={blue},urlcolor={red},breaklinks=true,final}{hyperref}
\PassOptionsToPackage{final}{graphicx}
\PassOptionsToPackage{final}{listings}

\documentclass[sigconf,nonacm]{acmart}

\AtBeginDocument{%
  \providecommand\BibTeX{{%
    \normalfont B\kern-0.5em{\scshape i\kern-0.25em b}\kern-0.8em\TeX}}}

\copyrightyear{2019}
\acmYear{2019}
\setcopyright{acmlicensed}
\acmConference[PPDP'19]{PPDP'19: 21st International Symposium on Principles and Practice of Declarative Programming}{October 07--09, 2019}{Porto, Portugal}
\acmBooktitle{PPDP'19: 21st International Symposium on Principles and Practice of Declarative Programming, October 07--09, 2019, Porto, Portugal}
\acmPrice{15.00}
\acmDOI{10.1145/1122445.1122456}
\acmISBN{978-1-4503-9999-9/18/06}

\usepackage{amsthm}
\newtheorem{remark}{Remark}

\usepackage{proof}
\usepackage{enumitem}
\usepackage{mathtools}

\usepackage[draft]{ifdraft}

\makeatletter
\@draftfalse
\@option@draftfalse
\makeatother

\ifdraft
{
  \usepackage[layout=footnote,draft]{fixme}
  \usepackage[notcite,notref,final]{showkeys}
  
}{
  \usepackage[layout=footnote,final]{fixme}
}

\FXRegisterAuthor{sg}{asg}{SG}	
\FXRegisterAuthor{rn}{arn}{RN}	

\newcommand{\dom}{\oname{dom}}
\newcommand{\cc}{{\operatorname{\mathsf{cc}}}}
\newcommand{\cd}{{\operatorname{\mathsf{cd}}}}
\newcommand{\od}{{\operatorname{\mathsf{od}}}}
\newcommand{\dr}{{\operatorname{\mathsf{dr}}}} 
\renewcommand{\tt}{{\operatorname{\mathsf{tt}}}} 
\newcommand{\ev}{{\operatorname{\mathsf{ev}}}} 
\newcommand{\hc}{\textsc{Hyb\-Core}}
\newcommand{\ie}{i.e.}
\newcommand{\eg}{e.g.}
\renewcommand{\VDash}{\vdash}
\newcommand{\ctx}[1]{\VDash_{#1}}
\newcommand{\cctx}{\ctx{\mathsf{c}}}
\newcommand{\vctx}{\ctx{\mathsf{v}}}

\definecolor{rules}{rgb}{0.8,0.4,0.23}
\definecolor{propositions}{rgb}{0.46,0.40,0.9}
\definecolor{assertions}{rgb}{0.10,0.5,0.36}

\DeclareRobustCommand{\mif}[3]{\mfix{\ruco{if}\,}{#1}{\,\ruco{then}\,}{#2}{\,\ruco{else}\,}{#3}}
\DeclareRobustCommand{\mwhile}[3]{\mfix{}{#1}{\,\ruco{while}\,}{#2}{\,}{\{#3\}}}
\DeclareRobustCommand{\cwhile}[2]{\mfix{}{#1\,}{\asco{\&}\,}{#2}}
\DeclareRobustCommand{\pcase}[2]{\mfix{}{#1;}{}{#2}}
\DeclareRobustCommand{\mbind}[2]{\mfix{\!}{#1;}{}{#2}}

\newcommand{\now}[1]{\lceil\kern-4pt\lceil #1 \rceil\kern-4pt\rceil}
\newcommand{\ass}{\mathrel{\asco{\coloneqq}}}
\newcommand{\integ}{\vaco{\mathit{int}}}
\newcommand{\wait}{\vaco{\mathit{wait}}}
\newcommand{\ruco}{\textcolor{rules}}

\newcommand{\vaco}{\textcolor{magenta}}

\newcommand{\asco}{\textcolor{assertions}}
\newcommand{\dist}{\oname{dist}}
\newcommand{\To}{\Downarrow} 

\DeclareFontEncoding{LS1}{}{}
\DeclareFontSubstitution{LS1}{stix}{m}{n}
\DeclareSymbolFont{arrows1}{LS1}{stixsf}{m}{n}
\DeclareMathDelimiter{\Ddownarrow}{\mathrel}{arrows1}{"60}{arrows1}{"60}
\newcommand{\TTo}{\mathrel{\scalebox{.6}[1.05]{$\Ddownarrow$}}}

\usepackage{bm}
\newcommand{\True}{\ruco{\oname{true}}}
\newcommand{\False}{\ruco{\oname{false}}}
\newcommand{\BBM}{{\bm{\mathsf{M}}}}
\newcommand{\BBH}{{\bm{\mathsf{H}}}}
\newcommand{\BBQ}{{\bm{\mathsf{Q}}}}
\newcommand{\BBT}{{\bm{\mathsf{T}}}}
\newcommand{\nat}{\boldsymbol{\BBN}}
\newcommand{\real}{\boldsymbol{\BBR}}
\newcommand{\realp}{\real_{\mplus}}
\newcommand{\realpe}{\bar\real_{\mplus}}
\newcommand{\ite}[3]{\mfix{\kern-2pt}{#1}{\lhd}{#2}{\rhd}{\mathbin{}#3}}
\newcommand{\klstar}{\star}  				
\newcommand{\istar}{\dagger}  				
\newcommand{\iistar}{\ddagger}  			
\newcommand{\out}{\operatorname{\mathsf{out}}}
\newcommand{\tuo}{\operatorname{\out^{\text{\kern.5pt\rmfamily-}\kern-.5pt1}\kern-1pt}}

\newcommand{\lrule}[3]{\textbf{#1}\quad\frac{#2}{#3}}
\newcommand{\anonrule}[3]{\infrule{#2}{#3}}
\newcommand{\infrule}[2]{\frac{#1}{#2}}
\newcommand{\none}{} 

\renewcommand{\div}{\circleddash} 

\newcommand{\ssto}[1]{
\newdimen\stringwidth
\setbox0=\hbox{\tiny$#1$}
\stringwidth=\wd0
    \mathrel{\raisebox{-.1ex}{\kern3pt\ensuremath{\mathrel{\tikz{\draw [-stealth,line width=0.4] (0.6ex,.1ex) -- node[font=\tiny, midway,above=-.3ex,xshift=-.2ex] {$#1$\kern2pt\erule} ++([xshift=2ex]\stringwidth,0) ;}
  }}\kern3pt}}
}

\newcommand{\sssto}[1]{
\newdimen\stringwidth
\setbox0=\hbox{\tiny$#1$}
\stringwidth=\wd0
    \mathrel{\raisebox{-.05ex}{\kern3pt\ensuremath{\mathrel{\tikz{\draw [-stealth,line width=0.4,double] (0.6ex,.1ex) -- node[font=\tiny, midway,above=-.3ex] {$#1$\kern2pt\erule} ++([xshift=2ex]\stringwidth,0) ;}
  }}\kern3pt}}
}

\usepackage{subcaption}

\input{catprog}

\renewcommand{\comp}{\,}

\newcommand{\ball}{\vaco{\mathit{ball}}}
\newcommand{\lin}{\vaco{\mathit{line}}}
\newcommand{\accel}{\vaco{\mathit{accel}}}
\newcommand{\brake}{\vaco{\mathit{brake}}}

\usepackage{tikz}

\usepackage{pgfplots}
\usepgfplotslibrary{patchplots}
\usepgfplotslibrary{fillbetween}
\usetikzlibrary{snakes}
\usetikzlibrary{patterns}
\usetikzlibrary{automata, positioning, arrows}
\usetikzlibrary{arrows}
\usepackage{wrapfig}

\begin{document}\allowdisplaybreaks

%
\title{An Adequate While-Language for Hybrid Computation}

%
\author{Sergey Goncharov}
\email{sergey.goncharov@fau.de}
\affiliation{%
  \institution{Friedrich-Alexander Universit\"at Erlangen-N\"urnberg, Germany}
}

\author{Renato Neves}
\email{nevrenato@di.uminho.pt}
\affiliation{%
  \institution{INESC TEC (HASLab) \& University of Minho, Portugal}
}

%
\renewcommand{\shortauthors}{Goncharov and Neves}

%
\begin{abstract}
  Hybrid computation harbours discrete and continuous dynamics in the
  form of an entangled mixture, inherently present in various natural
  phenomena and in applications ranging from control theory to
  microbiology.
  The emergent behaviours bear signs of both computational and
  physical processes, and thus present difficulties not only in their
  analysis, but also in describing them adequately in a structural,
  well-founded way.

  In order to tackle these issues and, more generally, to investigate
  \emph{hybridness} as a dedicated computational phenomenon, we
  introduce a while-language for hybrid computation inspired by the
  \emph{fine-grain call-by-value} paradigm. We equip it with
  \emph{operational} and \emph{computationally adequate denotational}
  semantics.  The latter crucially relies on a \emph{hybrid monad}
  supporting an (Elgot) iteration operator that we developed
  elsewhere. As an intermediate step, we introduce a more lightweight
  \emph{duration semantics} furnished with analogous results and based
  on a new \emph{duration monad} that we introduce as a lightweight
  counterpart to the hybrid monad.
\end{abstract}

%
%

%

%
\maketitle

\section{Introduction}

\textbf{\emph{Motivation and context.~}} Hybrid systems are
traditionally seen as computational devices that closely interact with
physical processes, such as movement, temperature, energy, and
time. Their wide range of applications, including \eg\ the design of
embedded systems \cite{Platzer08} and analysis of disease propagation
\cite{liu17}, led to an extensive research in the area, with a
particular focus on languages and models that suitably accomodate
classic computations and physical processes at the same
time~\cite{Henzinger96,Platzer08,BromanLeeEtAl12}. Most of this
research revolved around the notion of \emph{hybrid automaton}, widely
adopted in control theory and in computer
science~\cite{Henzinger96,AlurCourcoubetisEtAl93}, and around certain
classes of \emph{(weak) Kleene algebra}, specialised at verifying
safety and liveness properties of hybrid
systems~\cite{Platzer08,HofnerMoller09}.

A distinctive aspect of these models is the free use of non-determinism,
which semantically amounts to \emph{non-deterministic hybrid systems} rather
than \emph{(purely) hybrid systems}, the latter abundantly found in the wild.
Furthermore, iterative computations in these models are captured by Kleene star
(essentially the non-deterministic choice between all possible finite iterates
of the system at hand), i.e.\ non-determinism is also intertwined with recursion. In this
paper, we take a different view: we see non-determinism as an \emph{abstraction
layer} added on top of \emph{hybridness} for handling uncertainty. This
consideration can be made precise by drawing on Moggi's seminal work~\cite{Moggi91a}
associating \emph{computational effects} with \emph{(strong) monads}.
Both non-determinism and hybridness are computational effects in Moggi's
sense: while the former effect is standard, the latter was identified
recently~\cite{NevesBarbosaEtAl16} and elaborated in the second author's PhD
thesis~\cite{Neves18}.

Our view of (pure) hybridness as a computational effect connects the
area of hybrid systems with a vast body of research on monad-based
programming theory (\eg\ \cite{Moggi91a,LevyPowerEtAl02, Elgot75}),
which potentially eases the development of sophisticated \emph{hybrid
programming languages} -- a task that arguably belongs to the core of
the twenty-first century's technology, due to the increasing presence
of software products that interact with physical processes
\cite{BromanLeeEtAl12,KimKumar12}.
As a step in the direction of such languages, we introduce a
deterministic while-language for hybrid computation~\hc{}. We do not
intend to use it directly in the development of complex hybrid
systems, but rather as a basic language on which to study the subject
of \emph{(purely) hybrid computation}, possible extensions of the
latter, and as a theoretical basis for the development of more complex
hybrid programming languages.

\noindent
\textbf{\emph{Contributions and overview of \hc{}.~}} In the
development of \hc{}, we concentrate on the indispensable features of
hybrid computation, namely classic program constructs and primitives
for describing behaviours of physical processes. As we will see, the
orchestration of these two aspects yields significant challenges,
centring on the main construct of \hc{} -- the \emph{hybrid
  while-loop} -- which captures iterated hybrid behaviour, including
both finite and infinite loop unfoldings, most notably \emph{Zeno
  behaviour}.
\begin{figure*}[t]
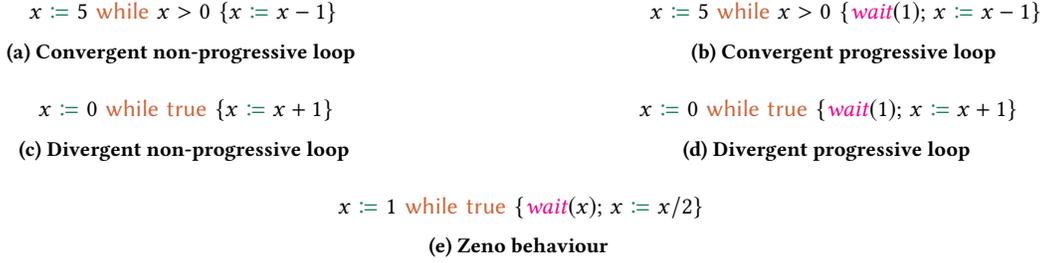

{\parbox{\textwidth}{%
 \centering
 \begin{subfigure}[b]{0.45\textwidth}
   \centering
  $\mwhile{x\ass 5}{x>0}{x\ass x-1}$
 \caption{Convergent non-progressive loop}\label{fig:while-a}
 \end{subfigure}
\qquad
  \begin{subfigure}[b]{0.46\textwidth}
  \centering
  $\mwhile{x\ass 5}{x>0}{\mbind{\wait(1)}{x\ass x-1}}$
 \caption{Convergent progressive loop}\label{fig:while-b}
 \end{subfigure}
\vspace{3ex}

 \begin{subfigure}[b]{0.42\textwidth}
  \centering
  $\mwhile{x\ass 0}{\True}{x\ass x+1}$
 \caption{Divergent non-progressive loop}\label{fig:while-c}
 \end{subfigure}
\qquad
  \begin{subfigure}[b]{0.46\textwidth}
  \centering
  $\mwhile{x\ass 0}{\True}{\mbind{\wait(1)}{x\ass x+1}}$
 \caption{Divergent progressive loop}\label{fig:while-d}
 \end{subfigure}

\vspace{3ex}
  \begin{subfigure}[b]{0.48\textwidth}
  \centering
  $\mwhile{x\ass 1}{\True}{\mbind{\wait(x)}{x\ass x/2}}$
 \caption{Zeno behaviour}\label{fig:while-e}
 \end{subfigure}
 \caption{Taxonomy of hybrid while-loops}\label{fig:whiles}
}}

\end{figure*}
Figure~\ref{fig:whiles} shows a taxonomy of while-loops emerging from
\hc{}:  besides the standard possibilities of \emph{convergence} and
\emph{divergence}, we now distinguish between \emph{progressive},
\emph{non-progressive}, and Zeno behaviour, depending on the
\emph{execution time of computations} (in Figure~\ref{fig:whiles} the
$\wait(s)$ command, triggers a delay of $s$ time units).

We emphasise the distinction between our semantics and the more
widespread \emph{transition semantics} for hybrid
automata~\cite{Henzinger96}.  Following Abramsky~\cite{Abramsky14},
the latter can be qualified as an \emph{intensional semantics}, while
we aim at an \emph{extensional semantics}, meaning that \eg\ we regard
the two programs
\begin{flalign*}
  \mwhile{x\ass 1}{\True}{\wait(1)} &\text{\quad and\quad } \mwhile{x\ass 1}{\True}{\wait(2)} 
\end{flalign*}
as equivalent, while they would be distinct under the transition
semantics.
%
%
In a nutshell, by committing to the extensional view, we abstract away
from specific choices of control states -- a decision that could also
be interpreted in hybrid automata setting and thus potentially lead to
coarser, more flexible notions of automata equivalence.

One principal design decision we adopt in~\hc{} is to base it on the
\emph{fine-grain call-by-value} paradigm~\cite{LevyPowerEtAl02} -- a
refinement of Moggi's computational $\lambda$-calculus, suitable for
both operational and denotational semantics. Denotationally, we
associate a hybrid computation with a \emph{trajectory}, i.e.\ a map
from a downward closed set~$D$ of non-negative real numbers (with
$D=\emptyset$ for the \emph{empty trajectory} representing
non-progressive divergence). 
The supremum of~$D$ is called the \emph{duration} of the hybrid
computation.  As shown in previous work~\cite{GoncharovJakobEtAl18,
  GoncharovJakobEtAl18a}, such trajectories constitute a \emph{hybrid
  monad}~$\BBH$. 

We introduce an operational semantics for \hc\ and show that it is
sound and (computationally) adequate w.r.t.~$\BBH$. This result relies
on a more lightweight operational semantics that keeps track of
durations (\ie\ execution times) only. We complement this latter
semantics with a \emph{duration monad} and prove analogous soundness
and adequacy results for it.

Our results do not fit standard semantic frameworks for the following reasons:
\begin{itemize}[wide]
\item It is not possible to introduce iteration for the duration monad
  (as well as for the hybrid monad ~\cite{GoncharovJakobEtAl18,
    GoncharovJakobEtAl18a}) via the usual least fixpoint argument,
  essentially due to the lack of canonical choice of divergence; the
  strategy we employ here, is thus to introduce a certain
  \emph{coalgebraic} version of this monad equipped with a partial
  (guarded) Elgot iteration and then to suitably quotient it.
\item Progressive divergence and Zeno behaviour, are principally
  \emph{non-inductive} kinds of behaviour, which we manage to capture
  by allowing for infinitely many premises in the operational
  semantics rules.
\item Our semantics for \hc\ is essentially two-stage, and contains
  the aforementioned duration semantics as an auxiliary
  ingredient. The operational semantics rules thus feature two
  separate evaluation judgements.
\end{itemize}

\noindent
\textbf{\emph{Related work.~}}
In designing the different constructs of \hc{}, we owe inspiration to
Witsenhausen's model of differential equations indexed by operations
modes, a percursor of the concept of hybrid system
\cite{witsenhausen66}, the \emph{relational} Kleene algebra underlying
\emph{differential dynamic logic}~\cite{Platzer08} and also hybrid
automata~\cite{Henzinger96}. As already mentioned, we restrict
ourselves to a purely hybrid setting, in particular we exclude
non-determinism. This is not to say that we could not have considered
non-determinism (in fact, there is also a non-deterministic hybrid
monad \cite{DahlqvistNeves18}) but for the time being we focus on pure 
hybrid behaviour in isolation from other computational effects, such as 
non-determinism.

Suenaga and Hasuo~\cite{SuenagaHasuo11} develop a somewhat similar
while-language for hybrid systems but from a different semantic
viewpoint: by unifying both continuous and discrete cyclic behaviour
in a single while-construct that makes incrementations by
infinitesimals. An interesting, principled approach to the semantics
of \emph{timed systems}, via an \emph{evolution comonad}, is proposed
by Kick, Power and Simpson~\cite{KickPowerEtAl06}. The underlying
functor of the evolution comonad is striking close to the functor
underlying our hybrid monad, except that the former does not include
empty trajectories and hence does not support non-progressive
divergence.

As indicated above, in order to interpret hybrid while-loops, we
necessarily went beyond standard semantic
frameworks~\cite{Winskel93,Reynolds98}, in particular because of the
indispensable presence of various notions of divergence. This relates
our work to the recent work on modelling productively non-terminating
programs using a coalgebraic \emph{delay monad} (previously called the \emph{partiality
  monad})~\cite{Capretta05,Danielsson12}. Indeed, our construction of
the iteration operator on the \emph{duration monad} in
Section~\ref{sec:Q-sem}, is inspired by a technically similar
procedure of quotienting the delay monad by weak
bisimilarity~\cite{ChapmanUustaluEtAl15}, albeit in a different
(constructive) setting.

\section{Preliminaries and Notation}\label{sec:prelim}
We will generally work in the category $\Set$ of sets and functions,
but also call on some standard idioms of category
theory~\cite{MacLane71,Awodey10}.  
By~$|\BC|$ we refer to the objects of~$\BC$, and by 
$\Hom_{\BC}(A,B)$ (or~$\Hom(A,B)$, if no
confusion arises) to the morphisms $f\colon A\to B$ from $A\in |\BC|$ to
$B\in|\BC|$. We often omit indices at natural
transformations. 
By~$\nat$,
$\real$, $\realp$ and~$\realpe$ we denote 
natural numbers (including~$0$), real numbers, non-negative real
numbers and extended non-negative real numbers $\realp\cup\,\{\infty\}$ respectively.
By $X+Y$ we denote a coproduct of~$X$ and~$Y$, the corresponding
coproduct injections $\inl\colon X\to X+Y$ and~$\inr\colon Y\to X+Y$, and dually for products: 
$\fst\colon X\times Y\to X$, $\snd\colon X\times Y\to Y$.  
%
A category with finite products and coproducts 
is \emph{distributive}~\cite{Cockett93} if the natural transformation
%
$[\id\times\inl,\id\times\inr]\colon X\times Y + X\times Z \to X\times (Y+Z)$
is an isomorphism, whose inverse we denote by $\dist$. By
$\ite{p}{b}{q}$ we abbreviate the if-then-else operator expanding as follows:
$\ite{p}{\top}{q} = p$, $\ite{p}{\bot}{q} = q$. For time-dependent functions 
$e\colon\real\to Y$ we use the superscript notation $e^t$ in parallel with $e(t)$ to 
emphasize orthogonality of the temporal dimension to the spatial one.

\begin{figure*}[h]
\begin{minipage}{0.85\textwidth}%
%
%
  \begin{flalign*}
    \lrule{(unit)}{}{\Gamma\vctx\star\colon 1}
    &&
    \lrule{(var)}{
		  x \colon A \text{~~in~~} \Gamma
	  }{
		  \Gamma \vctx x \colon A
	  }
    &&
    \lrule{(sig)}{
		  \vaco{f}\colon A\to B\in\Sigma\qquad\Gamma\vctx v\colon A
	  }{
		  \Gamma \vctx f(v) \colon B
	  }
  \end{flalign*}\\[-3.5ex]
  \begin{flalign*}
    \lrule{(true)}{}{\Gamma\vctx\True\colon 2}
    &&
    \lrule{(false)}{}{\Gamma\vctx\False\colon 2}
    &&
    \lrule{(prod)}{
		  \Gamma\vctx v\colon A\qquad \Gamma\vctx w\colon B
	  }{
		  \Gamma \vctx \brks{v,w} \colon A\times B
	  }\\[-3.5ex]
  \end{flalign*}
  \smallskip
  \dotfill 
  \smallskip
%
%
%
\vspace{1ex}
%
  \begin{flalign*}
    \lrule{(tr)}{
      \Gamma,t\colon\real\vctx v\colon A \qquad \Gamma,x\colon A\vctx w\colon 2
    }{
      \Gamma \cctx \cwhile{x\ass t.\,v}{w} \colon A
    }
    &&
    \lrule{(seq)}{
      \Gamma \cctx p\colon A \qquad \hspace{-0.1cm}
      \Gamma, x\colon A \cctx q \colon B
    }{
      \Gamma \cctx \mbind{x\ass p}{q} \colon B
    }
  \end{flalign*}\\[-3.5ex]
%
\begin{flalign*}
   \lrule{(pm)}{
      \Gamma \vctx v\colon A\times B \qquad
      \Gamma,x\colon A,y\colon B \cctx p \colon C
    }{
      \Gamma \cctx \pcase{\brks{x,y} \ass v}{p} \colon C
    }
&&
    \lrule{(now)}{
      \Gamma\vctx v\colon A
    }{
      \Gamma \cctx \now v \colon A
    }
\end{flalign*}\\[-3.5ex]
\begin{flalign*}
    \lrule{(if)}{
      \Gamma \vctx v\colon 2 \quad
      \Gamma \cctx p \colon A \quad \Gamma \cctx q \colon A
    }{
      \Gamma \cctx \mif{v}{p}{q} \colon A
    }
\hspace{-1cm}
&&
    \lrule{(wh)}{
		  \Gamma\cctx p\colon A \quad \Gamma,x\colon A\vctx w\colon 2\quad \Gamma,x\colon A\cctx q\colon A
	  }{
		  \Gamma\cctx \mwhile{x\ass p}{w}{q} \colon A
	  }
\end{flalign*}
\end{minipage}
\caption{\hc{}'s term formation rules for values (top) and computations (bottom)}
\label{fig:lang.full}
\end{figure*}

Following Moggi~\cite{Moggi91a}, we identify
a \emph{monad} $\BBT$ on a category $\BC$ with the corresponding
\emph{Kleisli triple} $(T,\eta,(\argument)^\klstar)$
consisting of an endomap~$T$ on $|\BC|$, a $|\BC|$-indexed class of
morphisms $\eta_X\colon X\to TX$, called the \emph{unit} of $\BBT$, and the
\emph{Kleisli lifting} maps
$(\argument)^\klstar\colon\Hom(X,TY)\to\Hom(TX,TY)$ such that
$
%
\eta^{\klstar}=\id,
f^{\klstar}\eta=f,
(f^{\klstar} g)^{\klstar}=f^{\klstar}g^{\klstar}.
$
Provided that $\BC$ has finite products, a
monad $\BBT$ on $\BC$ is \emph{strong} if it is equipped with 
\emph{strength}, i.e.\ a natural transformation
  $\tau_{X,Y}\colon X\times TY\to T(X\times Y)$
satisfying a number of standard coherence conditions
(see \eg~\cite{Moggi91a}). Every monad~$\BBT$ on $\Set$ is
strong~\cite{Kock72} under 
$\tau_{X,Y} = \lambda\brks{x,p}.\, T(\lambda y.\, \brks{x,y}) (p)$.
Morphisms of the form $f\colon X\to TY$ form the
\emph{Kleisli category} of $\BBT$, which has the same objects
as~$\BC$, units $\eta_X\colon X\to TX$ as identities, and composition
$(f,g)\mapsto f^\klstar g$, also called~\emph{Kleisli composition}.

An \emph{($F$-)coalgebra} on $X\in |\BC|$ for a functor $F\colon\BC\to\BC$
is a pair $(X,f\colon X\to FX)$. Coalgebras form a category whose morphisms
$h\colon (X,f)\to (Y,g)$ satisfy $h\in\Hom_{\BC}(X,Y)$ and $(Fh) f = g
h$. A final object $(\nu F,\out\colon\nu F\to F\nu F)$ of this category is
called a \emph{final coalgebra} (see~\cite{UustaluVene99} for more
details on coalgebras in semantics of (co)iteration).

\section{A While-Language for Hybrid Computation}
\label{sec:hybcore}
We start off by defining our core language for hybrid computation
\hc{}, following the \emph{fine-grain call-by-value}
paradigm~\cite{LevyPowerEtAl02}.  First, we introduce \emph{types}: 
%
\begin{align}\label{eq:sum-types}
A,B,\ldots          &\Coloneqq \nat\mid\real\mid 1\mid 2 \mid A\times B
\end{align}
and then postulate a signature $\Sigma$ of typed operation symbols of the
form $f\colon A\to B$ with $B\in\{\nat,\real,2\}$ for integer and real
arithmetic, \eg\ summation $+\colon\nat\times\nat\to\nat$, and further
primitives meant to capture time-dependent solutions of systems of
ordinary differential equations: \ie\ for every such system
$\{\dot x_1 = f_1(\bar x,t)\comma\ldots\comma \dot x_n = f_n(\bar x,t)\}$
where $\bar x$ is the vector $\brks{x_1,\ldots,x_n}$ and
$f_i \colon\real^n \times \real \to\real$, the signature~$\Sigma$ contains the symbol
$\integ_{\vaco{f_1,\ldots,f_n}}^{\vaco i}\colon \real^n\times\real\to
\real$, representing the $i$-th projection of the corresponding unique
(global) solution~\cite{Perko13}. 
\begin{remark}
We do not study conditions ensuring existence of the solutions 
$\integ_{\vaco{f_1,\ldots,f_n}}^{\vaco i}$, and instead work under the assumption 
that the functions~$f_i$ come from a sufficiently well-behaved class for
which these solutions exist. Clearly, we could instead introduce a specific
grammar of differential equations and make them part of the language. For 
example, we could use the following grammar for the~$f_i(\bar x,t)$:
\begin{align*}
  f(\bar x,t),g(\bar x,t),\ldots \Coloneqq&\\*
 c\in\real\mid t\mid x& \in \bar x\mid f(\bar x,t) + g(\bar x,t) \mid f(t)\cdot g(\bar x,t). 
\end{align*}
%
where $f(t)$ reads as $f(\bar x,t)$ with empty $\bar x$.
This corresponds to systems of
\emph{linear ordinary differential equations}, which provide a
sufficiently large stock of trajectories standardly used in hybrid
system modelling~\cite{Alur15}. A concrete choice of such grammar
however would have no bearing on our present results, but we expect to
revisit this decision in future work devoted to reasoning and
verification in \hc{}.
\end{remark}

\noindent
Our language features two kinds of judgement,
\begin{align}\label{eq:judg}
\Gamma\vctx v\colon A  \text{\qquad and\qquad } \Gamma \cctx p\colon A
\end{align}
for \emph{values} and \emph{computations}, respectively. These involve
\emph{variable contexts} $\Gamma$ which are non-repetitive lists of
\emph{typed variables} $x\colon A$. We indicate the \emph{empty context} as
$-$ (dash), \eg\ $-\cctx p\colon A$.  The term language over these data is
given in Figure~\ref{fig:lang.full}. Programs of \hc{} figure as
terms~$p$ in computation judgements $\Gamma\cctx p\colon A$ and are
inductively built over value judgements. Intuitively, the latter
capture functions from~$\Gamma$ to~$A$ in the standard mathematical
sense, and computation judgements capture functions from $\Gamma$ to
\emph{trajectories} valued on $A$; \ie\ instead of single values
from $A$, they return \emph{functions} $D \to A$ valued on $A$ with
$D=[0,d]$ or $D=[0,d)$, \eg\ $\emptyset = [0,0)$ and
$\realp = [0,\infty)$. Most of the constructs in
Figure~\ref{fig:lang.full} are standard.
Sequential composition \textbf{(seq)} for
example reads as ``bind the output of~$p$ to~$x$ and then
feed it to $q$''.  
The rule~\textbf{(now)} converts a value~$v$ into an instantaneous
computation, $[0,0] \to A$, that returns $v$.  The rules~\textbf{(if)}
and~\textbf{(wh)} provide basic control structures
(cf.~\cite{Winskel93}). Finally,~\textbf{(tr)} converts a value into a
computation by simultaneously abstracting over time and restricting
the duration of the computation to the \emph{union} of all those
intervals $[0,d]$ on which~$v$ holds throughout.  Note that the
resulting interval can thus be either open or closed on the right.
Intuitively, one can see~\textbf{(tr)} as a `continuous' while-loop
since the trajectory captured by $t.\,v$ runs as long as the condition
$w$ is satisfied.  For example, the computation judgement
$\Gamma \cctx \cwhile{x\ass t.\,t}{t \leq 7} \colon\real$ models the
passage of time and runs whilst the latter does not surpass seven time
units.  This behaviour captures the essence of hybrid behaviour 
\cite{witsenhausen66}, which is now accommodated by hybrid automata and 
Platzer's differential dynamic logic; our notation `$\&$' alludes precisely 
to this corresponding operator of the latter framework.

%

In contrast to fine-grain call-by-value~\cite{LevyPowerEtAl02}, we
have no conversions from computations to values, for we do not assume
existence of higher order types, which potentially provide a
possibility to \emph{suspend} a program as a value and subsequently
\emph{execute} it,
resulting in a computation. Adding these facilities to~\hc{} can be done standardly,
but we refrain from it, as we intend our semantics  to be eventually
transferable from $\Set$ to more suitable universes, which need not be
Cartesian closed, such as the category of topological spaces~$\Top$.

Next we introduce some syntactic conventions for \hc\ programs:
\begin{citemize}
\item We write $\cwhile{x\ass t.\,v}{r}$ with
  $r \in \real$ for 
  $\mbind{\brks{y,t} \ass (\cwhile{x\ass t.\,\brks{v,t}\\}{\snd x\leq
      r})}{\now{y}}$ (assuming that the projection
  $\snd\colon A\times B\to B$ and the predicate $\leq$ are in
  $\Sigma$). Intuitively, this restricts the trajectory encoded by
  $t.\, v$ 
  to the interval $[0,r]$.
\item Sequential composition is right associative, that is, an
  iterated sequential composition in the form
  $x_1\ass p_1; \ldots; x_n\ass p_n; q$ should be parsed as
  $x_1\ass p_1; (\ldots; (x_n\ass p_n; q)\ldots)$.
\item We allow for programs of the form $x\ass p$ which abbreviate
  $x\ass p;\now {x}$.
\item We omit the brackets $\now{\argument}$ 
if they can be easily reconstructed. E.g.\ $x\ass x+1$ means $x:=\now{x+1}$ and thus $x:=\now{x+1};\now{x}$,
by the previous clause (e.g.\ this is used in Figure~\ref{fig:whiles}).
\end{citemize}

\noindent
Let us examine two examples of hybrid systems programmable
in \hc.

\begin{example}[Bouncing ball]\label{expl:ball}
%
Consider the system of differential equations
$\{ \dot{u} = v\comma \dot{v} = \mathsf{g} \}$ describing continuous
movement of a ball in terms of its \emph{height} $u$ and \emph{velocity} $v$ 
with~$\mathsf{g}$ being Earth's gravity. We postulate the corresponding solutions
$\ball_{\vaco{u}} \colon\real^2 \times \real \to \real$ and
$\ball_{\vaco{v}} \colon\real^2 \times \real \to \real$ in $\Sigma$ so that $\ball_{\vaco{u}}(x,y,t)$
and $\ball_{\vaco{v}}(x,y,t)$ are the height and the velocity of the ball at
time $t$, assuming $x$ and $y$ as the corresponding initial values 
(at time $0$). Let $\ball$ be $\brks{\ball_{\vaco{u}}, \ball_{\vaco{v}}}$, and then e.g.\
%
  $\cwhile{\brks{u,v} \ass t.\, \ball(5,0,t)}{u \geq 0}$
%
  encodes the behaviour of the ball starting from the moment when it
  is dropped from height~$5$ until it hits the ground. The \hc{}
  program describing the resulting bouncing behaviour with infinitely
  many iterations is as follows:

\bigskip
\hspace{.5cm}\begin{minipage}{\dimexpr\textwidth-3cm}
\begin{flushleft}
$\brks{u,v}\ass\brks{5,0}$   \\[.8ex]
$\oname{\ruco{while}}\> \True\> \{$ \\[.8ex]
$\begin{array}{@{\quad}r@{\;}l}
    \cwhile{\brks{u,v}  &\ass t.\, \ball(u,v,t)}{u \geq 0};\\[.8ex]
    \brks{u,v}          &\ass \brks{u, -0.5  v}
\end{array}$\\[.8ex]
$\}$
\end{flushleft}
\xdef\tpd{\the\prevdepth}
\end{minipage}\medskip

\noindent
The movement thus programmed is depicted in
Figure~\ref{fig:plotball}. Note that the height is decreasing, tending
to zero in the limit. This requires infinitely many iterations, but
the process is finite in terms of physical time, because each
iteration gets progressively shorter. This yields an example of Zeno
behaviour~\cite{GoncharovJakobEtAl18,HofnerMoller11,NakamuraFusaoka05},
which is often dismissed in the hybrid systems literature, as
justified by the fact that such behaviour is on the one hand not
computationally realisable, and on the other hand notoriously
difficult to analyse. Here, we commit ourselves to the task of
faithful modelling Zeno behaviour as an inherent feature of physical,
biological and other systems.
\end{example} 
\begin{example}[A simplistic cruise controller]
  The following program implements a simplistic cruise controller
  programmed in \hc{}: 

\bigskip
\hspace{.5cm}\begin{minipage}{\dimexpr\textwidth-3cm}
\begin{flushleft}
$\brks{u,v} \ass \brks{0,0}$\\[.8ex]
$\oname{\ruco{while}}\> \True \> \{$ \\[.8ex]
$\begin{array}{@{\quad}r@{\;}l}
\oname{\ruco{if}} \> v \leq 120~ \oname{\ruco{then}} \>
    \cwhile{\brks{u,v}  \ass& t.\, \accel(u,v,t)}{3}\\[.8ex]
    \oname{\ruco{else}} \> \cwhile{\brks{u,v}  \ass& t.\, \brake(u,v,t)}{3} 
\end{array}$\\[.8ex]
$\}$
\end{flushleft}
\xdef\tpd{\the\prevdepth}
\end{minipage}\bigskip

\noindent
Here, $\accel = \brks{\accel_{\vaco{u}},\accel_{\vaco{v}}}$ and
$\brake = \brks{\brake_{\vaco{u}}, \brake_{\vaco{v}}}$ refer to the
solutions of the systems of differential equations
$\{ \dot{u} = v, \dot{v} = 1 \}$, $\{ \dot{u} = v, \dot{v} = -1 \}$
correspondingly.  As the names suggest, $\accel$ is responsible for
the acceleration, while $\brake$ is responsible for braking. Each 
iteration in the loop has the duration of three
time units (due to the expression `$\asco{\&} \> 3$'); at
the beginning of each iteration the controller checks the current
velocity, leading the vehicle to the target velocity of $120$
km/h. This produces the oscillation pattern depicted in
Figure~\ref{fig:plotvehicle}.
\end{example}

\begin{example}[Signal sampling]\label{expl:sigsamp}
The program below discretises a continuous-time signal 
$\mathit{\vaco{signal}}\colon\real\times\real\to\real$.
The variable $u$, which is valued on natural numbers $\nat$, 
represents the sampled signal, and the operation $\mathsf{round}\colon\real\to\nat$ 
rounds the input up to the closest natural number.

\bigskip \hspace{.5cm}\begin{minipage}{\dimexpr\textwidth-3cm}
\begin{flushleft}
$\brks{u,v}\ass\brks{0,0}$   \\[.8ex]
$\oname{\ruco{while}}\> \True\> \{$ \\[.8ex]
$\begin{array}{@{\quad}r@{\;}l}
   \cwhile{\brks{u,v}  &\ass t.\, \brks{\mathsf{round}(v),\mathit{\vaco{signal}}(v,t)}}{1} 
\end{array}$\\[.8ex]
$\}$
\end{flushleft}
\xdef\tpd{\the\prevdepth}
\end{minipage}\medskip

\noindent
The resulting computation when projected to the second Cartesian
factor yields the original signal, while the same computation
projected to the first Cartesian factor yields the corresponding
discretised, or \emph{sampled} signal. Here, we used the unit $1$ as a
\emph{sampling interval}.
%
%
\end{example}

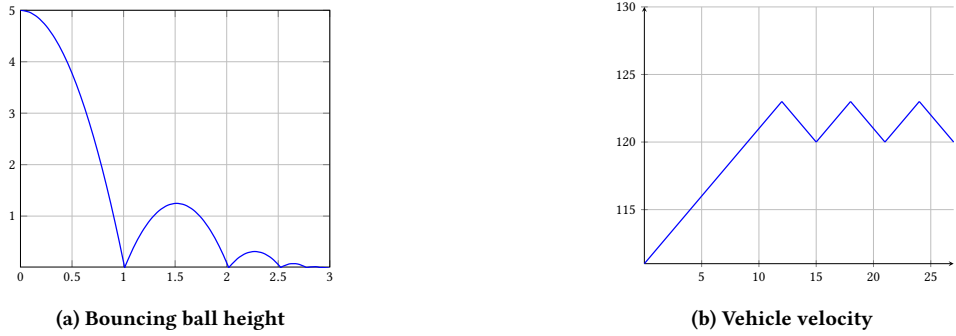
\begin{figure*}  \vspace{.25cm}
\centering
\hspace{.5cm}
 \begin{subfigure}[t]{0.42\textwidth}
  %
  %
  \begin{center}
     \scalebox{0.6}{
       \begin{tikzpicture}
         \begin{axis}[ymin= 0.01, ymax = 5, xmin = 0, xmax = 3, grid = major ]
           \addplot [thick,color=blue,smooth, domain = 0:2] { 5 - (1/2)*9.8*x^2 };
           \addplot[thick,color=blue,smooth, shift = {(101.0,0.0)}, domain = 0:2] {
             4.945*x - (1/2)*9.8*x^2 }; 
           \addplot[thick,color=blue,smooth, shift =
           {(202.0,0.0)}, domain = 0:2] { 2.473*x - (1/2)*9.8*x^2 };
           \addplot[thick,color=blue,smooth, shift = {(252.0,0.0)}, domain = 0:2] {
             1.2365*x - (1/2)*9.8*x^2 }; 
           \addplot[thick,color=blue,smooth, shift =
           {(278.0,0.0)}, domain = 0:2] { 0.62*x - (1/2)*9.8*x^2 };
           \addplot[thick,color=blue,smooth, shift = {(290.0,0.0)}, domain = 0:2] {
             0.31*x - (1/2)*9.8*x^2 };
         \end{axis}
       \end{tikzpicture}
     }
     \end{center}
    \caption{Bouncing ball height}\label{fig:plotball}
 \end{subfigure}
\hspace{0.5cm}
 \begin{subfigure}[t]{0.42\textwidth}
%
%
   \begin{center}
     \scalebox{0.6}{
       \begin{tikzpicture}
         \begin{axis}[ymin=111, ymax = 130, xmin = 0, xmax = 27, 
           axis x line = center,
           axis y line = center, 
           grid = major ]
           \addplot[thick,color=blue,smooth, domain = 0:12] { x + 111};
           \addplot[thick,color=blue,smooth, domain = 12:15] { 123 - (x - 12)};
           \addplot[thick,color=blue,smooth, domain = 15:18] { 120 + (x - 15)};
           \addplot[thick,color=blue,smooth, domain = 18:21] { 123 - (x - 18)};
           \addplot[thick,color=blue,smooth, domain = 21:24] { 120 + (x - 21)};
           \addplot[thick,color=blue,smooth, domain = 24:27] { 123 - (x - 24)};
         \end{axis}
       \end{tikzpicture}
     }
   \end{center}
\caption{Vehicle velocity}\label{fig:plotvehicle}
\end{subfigure}
\caption{Trajectories produced by hybrid programs}
\end{figure*}

\section{Operational Duration Semantics}
\label{sec:dur_op_sem}

We now introduce an \emph{operational duration semantics} of~\hc, both
in small-step and big-step styles. In order to help build intuitions,
we begin with the former. Duration semantics abstract away from the
run-time behaviour of the program (\ie\ its trajectory), and only
provide the execution time and possibly the final result (\ie\ the
endpoint of the trajectory). The small-step judgements are defined on
closed computations terms and have the form,
\begin{align*}
  p\ssto{d} q\qquad\text{and}\qquad p\ssto{e} && (d\in\real_{\mplus}, e\in\bar\real_{\mplus})
\end{align*}
meaning that $p$ one-step reduces to $q$ in time $d$, or $p$ one-step
diverges in time~$e$ respectively, where only in the latter case the
relevant time interval may be infinite. The derivation rules for these
judgements are shown in Figure~\ref{fig:small-dur-oper}.  These rules
are essentially obtained by extending the standard call-by-value
small-step semantics~\cite{Winskel93} by decorating the transitions
with durations.  The most distinctive rules are the ones addressing
the abstraction construct~\textbf{(tr)}:  intuitively, the resulting
durations $d$ are computed as the lengths of the largest interval
$[0,t]$ on which the trajectory produced by $v$ satisfies the
predicate $w$ throughout.  The terminal value $\now{v[d/t]}$ is
available precisely when the trajectory produced by $v$ satisfies the
predicate $w$ at $d$. For example, for
$\lin \colon\real \times \real \to \real$ the global solution of the
differential equation $\dot{x} = 1$, the program
$\cwhile{x \ass t.\, \lin(0,t)}{x \leq 1}$ one-step reduces to
$\now 1$ in time $1$, and on the other hand,
$\cwhile{x \ass t.\, \lin(0,t)}{x < 1}$ one-step diverges in time~$1$.

Note a certain arbitrariness in defining \emph{instantaneous
  transitions}~$\ssto{0}$. E.g.\ we could just as well use the
following slightly more economical rule for reducing
$\mif{\True}{p}{q}$:
\begin{displaymath}
  \anonrule{(if1)}{
    p\ssto{d} p'
  }{
    \mif{\True}{p}{q} \ssto{d} p'
  }
\end{displaymath}
This would however not change the principal fact of pervasiveness of intermediate
instantaneous transitions, as they are also triggered by unfolding while-loops,
and cannot be completely eliminated from the corresponding rules.

We proceed to introduce the aforementioned big-step duration
semantics. It associates a closed computation term $p$ either with a
pair $d,v$, consisting of a \emph{duration} $d\in\real_{\mplus}$ and a
(closed) \emph{terminal value} $v$, or with a (possibly infinite)
duration $d\in\bar\real_{\mplus}$. The latter case occurs \eg\ if $p$
exhibits Zeno behaviour or more generally if $p$ is a while-loop that
unfolds infinitely many times. The rules for big-step judgements are
given in Figure~\ref{fig:dur-oper}.

\begin{figure*}[t]
\begin{minipage}{0.90\textwidth}%
\textit{Closed values, Closed computations:}
\begin{align*}
   v,w &\Coloneqq x\mid \star\mid\True\mid\False\mid\brks{v,w}\mid f(v) \qquad (f\in\Sigma) \\
   p,q &\Coloneqq \pcase{\brks{x,y} \ass \brks{v,w}}{p}\mid\mif{v}{p}{q}\mid \now v
        \mid\mbind{x\ass p}{q}\mid\mwhile{x\ass p}{v}{q}\mid\cwhile{x\ass t.\,v}{w}
\end{align*}
\noindent
\textit{Rules:}
  \begin{flalign*}
    \anonrule{match1}{
    }{
      \pcase{\brks{x,y}\ass\brks{v,w}}{q} \ssto{0} q[v/x,w/y]
    }
&&
    \anonrule{ret-bind1}{
      p \ssto{d} p'
    }{
      \mbind{x\ass p}{q} \ssto{d} \mbind{x\ass p'}{q}
    }
&&
  \anonrule{ret-bind2}{
    p \ssto{d} \none
  }{
    \mbind{x\ass p}{q} \ssto{d} \none
  }
  \end{flalign*}\\[-3ex]
%
  \begin{flalign*}
    \anonrule{ret-bind2}{
    }{
      \mbind{x \ass \now v}{q} \ssto{0} q[v/x]
    }
&&
 \anonrule{(if1)}{
  }{
    \mif{\True}{p}{q} \ssto{0} p
  }
&&
\anonrule{(if2)}{
  }{
    \mif{\False}{p}{q} \ssto{0} q
  }
 \end{flalign*}
\\[-3ex]
  \begin{flalign*}
  \anonrule{(cond2)}{
      \forall s\leq d.\, w[v[s/t]/x]=\True\qquad \forall e>0.\,\exists s\in (d,d+e).\, w[v[s/t]/x]=\False
	}{
		\cwhile{x\ass t.\,v}{w} \ssto{d} \now{v[d/t]}
	}
&&
\anonrule{(d-while1)}{
      p\ssto{d} \none
	  }{
		  \mwhile{x\ass p}{v}{q} \ssto{d} \none
	  }
 \end{flalign*}
\\[-3ex]
  \begin{flalign*}
  \anonrule{(cond1)}{
      \forall s<d.\, w[v[s/t]/x]=\True\qquad w[v[d/t]/x]=\False
	}{
		\cwhile{x\ass t.\,v}{w} \ssto{d} \none
	}
&&
  \anonrule{(d-while2)}{
      w[v/x] = \False
	}{
		\mwhile{x\ass \now v}{w}{q} \ssto{0} \now v
	}
 \end{flalign*}\\[-3ex]
  \begin{flalign*}
    \anonrule{(d-while1)}{
        p \ssto{d} p'
	  }{
		  \mwhile{x\ass p}{v}{q} \ssto{d} \mwhile{x\ass p'}{v}{q}
	  }
  &&
  \anonrule{(d-while4)}{
      w[v/x] = \True
	}{
          \mwhile{x\ass \now v}{w}{q} \ssto{0} \mwhile{x\ass q[v/x]}{w}{q}
        }
\end{flalign*}
\end{minipage}      
\caption{Small-step duration operational semantics}\label{fig:small-dur-oper}
\end{figure*}

Again, we focus only on the least standard rules.  In
both~\textbf{(tr$_{\bm1}^{\mathbf{d}}$)}
and~\textbf{(tr$_{\bm2}^{\mathbf{d}}$)} the duration is identified
with the length of the largest interval $[0,d]$ on which the
condition~$w$ holds throughout, and subsequently each rule verifies
if~$v$ is true at~$t=d$.
The
rules~\textbf{(seq$_{\bm1}^{\mathbf{d}}$)},~\textbf{(seq$_{\bm2}^{\mathbf{d}}$)}
and~\textbf{(seq$_{\bm3}^{\mathbf{d}}$)} capture the fact that
durations of subsequent computations add up, unless the terminal value
of the first computation is undefined, in which case the total
duration of the program is the duration of the first subprogram. This
behaviour is inherited inductively by while-loops
via~\textbf{(wh$_{\bm1}^{\mathbf{d}}$)}--\textbf{(wh$_{\bm4}^{\mathbf{d}}$)},
and in addition we have a \emph{non-inductive case} of infinitely many
unfoldings of while-loops, as captured by the (infinitary)
rule~\textbf{(wh$_{\bm5}^{\mathbf{d}}$)}. This allows for computing
the duration of the program as the infinite sum~$\sum\nolimits_i d_i$.

\begin{remark}\label{rem:while}
  The rules~\textbf{(wh$_{\bm3}^{\mathbf{d}}$)}
  and~\textbf{(wh$_{\bm4}^{\mathbf{d}}$)} directly capture the
  inductive intuition in the semantics of while-loops, but notably
  these are the only rules which are not structural, in the sense that
  the programs in the premises are not structurally simpler than the
  corresponding programs in the conclusions.  It is easy to see
  that~\textbf{(wh$_{\bm3}^{\mathbf{d}}$)}
  and~\textbf{(wh$_{\bm4}^{\mathbf{d}}$)} can be equivalently replaced
  by the families of structural rule schemes presented in
  Figure~\ref{fig:rule:while} in the style
  of~\textbf{(wh$_{\bm5}^{\mathbf{d}}$)}.
\end{remark}

\begin{proposition}[Determinacy]\label{prop:deter}
  The semantics in Figure~\ref{fig:dur-oper} is deterministic, that is
  for every closed computation term $p$, there is (exclusively!) at most one
  judgement of the form $p\To d,v$ or $p\To d$ derivable using the
  rules of Figure~\ref{fig:dur-oper}.
\end{proposition}

\begin{proof}
The proof runs routinely by induction over the
number of computation constructs (bottom part of
Figure~\ref{fig:lang.full}) in $p$ and by inspecting the rules in
Figure~\ref{fig:dur-oper}, except
for~\textbf{(wh$_{\bm3}^{\mathbf{d}}$)}
and~\textbf{(wh$_{\bm4}^{\mathbf{d}}$)}, which must be replaced by
their structural analogues per Remark~\ref{rem:while}.   
\end{proof}

\noindent
In order to relate small-step with big-step semantics, we define
$p\sssto{d} q$ and $p\sssto{d}$ inductively by the following
small-step judgements and closure rules:
\begin{flalign*}
\infrule{}{\now{v}\sssto{0} \now{v}}&&
\infrule{p\ssto{d} q}{p\sssto{d} q}&&
\infrule{p\sssto{d} p'\quad p'\ssto{e} q}{p\sssto{d+e} q}&&
\infrule{p\ssto{d} }{p\sssto{d} }
\end{flalign*}
\vspace{-.5ex}
\begin{align*}
\infrule{p\sssto{d} p'\qquad p'\ssto{e}\none}{p\sssto{d+e}\none } &&
\infrule{p\ssto{d_1} p'\qquad p'\ssto{d_2} p''\quad \ldots}{p\sssto{d_1+d_2+\ldots}\none }
\end{align*}
The binary relation $\sssto{d}$ is then essentially the standard
transitive closure, modulo the fact that durations are added along the
way. The unary relation $\sssto{d}$ is derivable both from finite and
infinite numbers of premises, \eg\ in the Zeno behaviour case, as in
Example~\ref{expl:ball}.
\begin{theorem}\label{thm:duration-semantics}
For every closed computation term $p$, $p\sssto{d} \now{v}$ iff $p\To d,v$ and
$p\sssto{d} \none$ iff~$p\To d$.
\end{theorem}

The claimed equivalence is entailed by the following
Lemmas \ref{lem:finite_chains}\dash\ref{lem:eq-inf-op}.
\begin{lemma}\label{lem:finite_chains} For any two closed programs
  $p,q$,
\begin{enumerate}
\item unless~$p \neq \now v$, the judgement $p \sssto{d} q$ is
  derivable iff there exists a finite chain
  $p \ssto{d_1} \dots \ssto{d_n} q$ such that $\sum d_i = d$;
  \item the judgement $p\sssto{e}$ is derivable iff
  \begin{enumerate}[label=(\alph*)]
  \item either $p \ssto{d_1} p' \ldots \ssto{d_n} q\ssto{d_{n+1}}$
    (with $n=0$ meaning $p \ssto{e} $),
    \item or $p\ssto{d_1} p'\ssto{d_2} \ldots,$
  \end{enumerate}
  with suitable $p', \ldots$ and $q$ where $\sum\nolimits_i d_i = d$.
\end{enumerate}
\end{lemma}
\noindent
Note that $\now{v}$ is not one-step reducible, and hence
$\now{v} \sssto{0} \now{v}$ is the only provable judgement of the form
$\now{v} \sssto{d} q$.

\begin{lemma}\label{lem:eq-finite}
For every closed program $p$, $p\To d,v$ implies\ ${p\sssto{d} \now v}$.
\end{lemma}

\begin{proof}
  We proceed by induction over the derivation of $p\To d,v$. Note that
  every rule in Figure~\ref{fig:dur-oper}, whose conclusion is of the
  form $p\To d,v$, has as premise big-step judgements only the
  judgements in the same form. In order to obtain a proof, we thus
  need to check the induction step for every rule. This is
  straightforward for all the rules, except for the ones for
  sequential composition and for while-loops. We consider the ones
  that apply, in detail.

  {\textbf{(seq$_{\bm3}^{\mathbf{d}}$)}} The premises read as
  $p\To d,v$ and $q[v/x]\To e,w$, hence, by induction,
  $p\sssto{d} \now v$, $q[v/x]\sssto{e} \now w$.  If $p=\now v$ then
  $d=0$ and $\mbind{x\ass p}{q}\ssto{0} q[v/x]\sssto{e} \now w$, and
  hence $\mbind{x\ass p}{q}\sssto{d+e} \now w$, as desired. Otherwise,
  by Lemma~\ref{lem:finite_chains}, 
\begin{displaymath}
  \mbind{x\ass p}{q}\ssto{d_1}\ldots \ssto{d_n} \mbind{x\ass \now v }{q}\ssto{0} q[v/x] \sssto{e}  \now w
\end{displaymath}
where $p\ssto{d_1}\ldots \ssto{d_n} \now v$, $d=\sum\nolimits_i d_i$. Hence,
again $\mbind{x\ass p}{q}\sssto{\!d+e\!} \now w$.
{\textbf{(wh$_{\bm2}^{\mathbf{d}}$)}} The premises imply by induction
hypothesis $p \sssto{d} \now w$. The case $p=\now{w}$ is as in the previous clause. Otherwise,
we have
\begin{align*}
  \mwhile{x\ass p}{v}{q} \ssto{d_1}\ldots \ssto{d_n}& \mwhile{x\ass \now w}{v}{q}\\ \ssto{0}&\; \now w
\end{align*}
with $p\ssto{d_1}\ldots \ssto{d_n} \now w$, $d=\sum\nolimits_i d_i$, which yields
the derivation $\mwhile{{x\ass p}}{v}{q}\sssto{d}\now w$.

{\textbf{(wh$_{\bm4}^{\mathbf{d}}$)}} We obtain $p \sssto{d} \now w$
and $\mwhile{x\ass q[w/x]}{v}{q} \sssto{r} \now u$ from the rule
premises. After disregarding the obvious case $p=\now w$, we obtain
similarly to the previous clause
\begin{align*}
  \mwhile{x\ass p}{v}{q} \ssto{d_1}\ldots \ssto{d_n} \mwhile{x\ass \now w}{v}{q} \\
  \sssto{r} \now u
\end{align*}
Hence $\mwhile{x\ass p}{v}{q}\sssto{d+r} \now u$, and we are done.
\end{proof}

\begin{lemma}\label{lem:eq-finite-op}
  For every closed program $p$, $p\sssto{d} \now v$ implies
  $p\To d,v$.
\end{lemma}
\begin{proof}
  It suffices to show that $p\ssto{d} p'$ with $p'\To r,v$ imply
  $p \To d+r, v$, from which the claim follows by obvious
  induction. The latter statement in turn follows by induction over
  the derivation of $p\ssto{d} p'$ using the rules in
  Figure~\ref{fig:small-dur-oper}.
\end{proof}
We consider the proofs of the analogous results in the non-terminating
cases in more detail.

\begin{figure*}
\begin{minipage}{0.90\textwidth}%
\begin{flalign*}
  \lrule{(tr$_{\bm1}^{\mathbf{d}}$)}{
      \forall s\leq d.\, w[v[s/t]/x]=\True\qquad \forall e>0.\,\exists s\in (d,d+e).\, w[v[s/t]/x]=\False
	}{
		\cwhile{x\ass t.\,v}{w} \To d, v[d/t]
	}
&&
  \lrule{(pm$_{\bm1}^{\mathbf{d}}$)}{
    p[v/x,w/y] \To d
  }{
    \pcase{\brks{x,y}\ass\brks{v,w}}{p} \To d
  }
\end{flalign*}\\[-3ex]
\begin{flalign*}
  \lrule{(tr$_{\bm2}^{\mathbf{d}}$)}{
      \forall s<d.\, w[v[s/t]/x]=\True\qquad w[v[d/t]/x]=\False
	}{
		\cwhile{x\ass t.\,v}{w} \To d
	}
&&
  \lrule{(pm$_{\bm2}^{\mathbf{d}}$)}{
    p[v/x,w/y] \To d, u
  }{
    \pcase{\brks{x,y}\ass\brks{v,w}}{p} \To d, u
  }
\end{flalign*}\\[-3ex]
%
\begin{flalign*}
    \lrule{(seq$_{\bm1}^{\mathbf{d}}$)}{
      p \To d
    }{
      \mbind{x\ass p}{q} \To d
    }
&&
    \lrule{(seq$_{\bm2}^{\mathbf{d}}$)}{
      p \To d, v\qquad q[v/x]\To e
    }{
      \mbind{x\ass p}{q} \To d+e
    }
&&
    \lrule{(seq$_{\bm3}^{\mathbf{d}}$)}{
      p \To d, v\qquad q[v/x]\To e, w
    }{
      \mbind{x\ass p}{q} \To d+e, w
    }
\end{flalign*}\\[-3ex]
\begin{flalign*}
    \lrule{(now$^{\mathbf{d}}$)}{
    }{
      \now v \To 0, v
    }
&&
   \lrule{(if$_{\bm1}^{\mathbf{d}}$)}{
      p\To d
    }{
      \mif{\True}{p}{q} \To d
    }
&&
    \lrule{(if$_{\bm2}^{\mathbf{d}}$)}{
      q\To d
    }{
      \mif{\False}{p}{q} \To d
    }
\end{flalign*}\\[-3ex]
\begin{flalign*}
   \lrule{(if$_{\bm3}^{\mathbf{d}}$)}{
      p\To d, v
    }{
      \mif{\True}{p}{q} \To d,v
    }
&&
    \lrule{(if$_{\bm4}^{\mathbf{d}}$)}{
      q\To d, v
    }{
      \mif{\False}{p}{q} \To d,v
    }
&&
    \lrule{(wh$_{\bm1}^{\mathbf{d}}$)}{
      p \To d
    }{
      \mwhile{x\ass p}{v}{q} \To d
    }
\end{flalign*}\\[-3ex]
\begin{flalign*}
  \lrule{(wh$_{\bm2}^{\mathbf{d}}$)}{
      p \To d,w\qquad v[w/x] = \False
	}{
		\mwhile{x\ass p}{v}{q} \To d,w
	}
&&
  \lrule{(wh$_{\bm3}^{\mathbf{d}}$)}{
      p \To d,w\qquad v[w/x] = \True\qquad \mwhile{x\ass q[w/x]}{v}{q} \To r
  }{
	  \mwhile{x\ass p}{v}{q} \To d + r
  }
\end{flalign*}\\[-3ex]
%
\begin{flalign*}
  \lrule{(wh$_{\bm4}^{\mathbf{d}}$)}{
      p \To d,w\qquad v[w/x] = \True\qquad \mwhile{x\ass q[w/x]}{v}{q} \To r,u
	}{
		\mwhile{x\ass p}{v}{q} \To d + r, u
	}
  \end{flalign*}\\[-3ex]
\begin{flalign*}
&& 
   \lrule{(wh$_{\bm5}^{\mathbf{d}}$)}{
        p\To d_0,w_0\qquad 
        q[w_0/x] \To d_{1},w_{1}\qquad q[w_1/x] \To d_{2},w_{2}\qquad \ldots        
        \qquad \forall i\in\omega.\, v[w_i/x] = \True
	  }{
		  \mwhile{x\ass p}{v}{q} \To \sum\nolimits_i d_i
	  }&&
\end{flalign*}
\end{minipage}        
\caption{Big-step duration operational semantics}\label{fig:dur-oper}
\end{figure*}
\begin{lemma}\label{lem:eq-inf}
For every closed program $p$, $p\To d$ implies $p\sssto{d} \none\,$.
\end{lemma}
\begin{proof}
  The proof runs analogously to the proof of Lemma~\ref{lem:eq-finite}
  and occasionally calls the latter. Again, we restrict ourselves to
  verifying the rules for sequential composition and while-loops and
  make free use of Lemma~\ref{lem:finite_chains}.

  {\textbf{(seq$_{\bm1}^{\mathbf{d}}$)}} The premise of the rule
  implies $p\sssto{d}$ by induction hypothesis. We proceed by case
  distinction.
  \begin{itemize}
    \item $p\sssto{e} p'$ and $p'\ssto{r}$ where $d = e+r$. Then
  $\mbind{x\ass p}{q}\sssto{e} \mbind{x\ass p'}{q}\ssto{r}$, and hence
 $\mbind{x\ass p}{q}\sssto{d}$
    \item $p\ssto{d_1} p'\ssto{d_2}\ldots$, $d=\sum\nolimits_i d_i$. Then
\begin{displaymath}
  \mbind{x\ass p}{q}\ssto{d_1} \mbind{x\ass p'}{q}\ssto{d_2}\ldots,
\end{displaymath} and hence again
 $\mbind{x\ass p}{q}\sssto{d}$.
  \end{itemize}

  {\textbf{(seq$_{\bm2}^{\mathbf{d}}$)}} Using
  Lemma~\ref{lem:eq-finite} and the induction hypothesis, the rule
  premises imply $p \sssto{d} \now v$ and $q[v/x]\sssto{e}$. After
  disregarding the obvious case $p=\now v$, the former implies
  $\mbind{x\ass p}{q}\sssto{d} \mbind{x\ass\now v}{q}$, and therefore,
  the requisite reduction is constructed as follows:
\begin{displaymath}
  \mbind{x\ass p}{q}\sssto{d} \mbind{x\ass\now v}{q} \ssto{0} q[v/x]\sssto{e}.
\end{displaymath}

{\textbf{(wh$_{\bm1}^{\mathbf{d}}$)}} The induction hypothesis implies
$p\sssto{d}$. By case distinction:
  \begin{itemize}
    \item $p\sssto{e} p'$ and $p'\ssto{r}$ where $d = e+r$; then we obtain
  $\mwhile{x\ass p}{v}{q}\sssto{e} \mwhile{x\ass p'}{v}{q}\ssto{r}$, and hence
 $\mwhile{x\ass p}{v}{q}\sssto{d}$
    \item $p\ssto{d_1} p'\ssto{d_2}\ldots$, $d=\sum\nolimits_i d_i$; then
\begin{displaymath}
  \mwhile{x\ass p}{v}{q}\ssto{d_1} \mwhile{x\ass p'}{v}{q}\ssto{d_2}\ldots,
\end{displaymath} and hence again  $\mwhile{x\ass p}{v}{q}\sssto{d}$.
  \end{itemize}

  {\textbf{(wh$_{\bm3}^{\mathbf{d}}$)}} By applying the induction
  hypothesis and Lemma~\ref{lem:eq-finite} to the premises, we obtain
  the judgments $p \sssto{d} \now w$ and
  $\mwhile{x\ass q[w/x]}{v}{q} \sssto{r}$.  After disregarding the
  obvious case $p=\now w$, we obtain:
  \begin{align*}
    \mwhile{x\ass p}{v}{q} &\,\sssto{d} \mwhile{x\ass\now w}{v}{q}\ssto{0}\\
    \mwhile{x\ass q[w/x]}{v}{q}  &\,\sssto{r} 
  \end{align*}

  {\textbf{(wh$_{\bm5}^{\mathbf{d}}$)}} Finally, let us consider the
  case of Zeno behaviour of the while-loop. By applying
  Lemma~\ref{lem:finite_chains} to the rule premises, we obtain
  $p\sssto{d_0} \now {w_0}$, and
  $q[w_i/x]\sssto{d_{i+1}} \now{w_{i+1}}$ for all $i\in\omega$.  Thus
  we obtain the following infinite chain
\begin{align*}
\mwhile{x\ass p}{v}{q}
\sssto{d_0}&\; \mwhile{x\ass \now{w_0}}{v}{q}\ssto{0}\\ & \; \mwhile{x\ass q[w_0/x]}{v}{q}\\
\sssto{d_1}&\; \mwhile{x\ass \now{w_1}}{v}{q}\ssto{0}\\ & \; \mwhile{x\ass q[w_1/x]}{v}{q}\\
\vdots\quad&\;
\end{align*}

\vspace{-2ex}
\noindent as desired.
\end{proof}

\begin{figure*}
  \begin{minipage}{0.90\textwidth}\normalfont
  \begin{align*}
        \lrule{(wh${_{\bm3}^{\mathbf{d}}}{}'$)}{
    p \To d_0,w_0\qquad (q[w_i/x] \To d_{i+1},w_{i+1})_{i < n} \qquad
      q[w_n/x]  \To d_{n +1} \qquad \forall i \leq n.\, v[w_i/x] = \True
            }{
                    \mwhile{x\ass p}{v}{q} \To d_0 +\ldots+ d_{n +1 }
            }\qquad (n\in\omega)\\
  \end{align*}
  \begin{align*}
        \lrule{(wh${_{\bm4}^{\mathbf{d}}}{}'$)}{
         p \To d_0,w_0\qquad (q[w_i/x] \To d_{i+1},w_{i+1})_{i < n} \qquad
    v[w_n/x] = \False \qquad  \forall i< n.\, v[w_i/x]
    = \True
            }{
                    \mwhile{x\ass p}{v}{q} \To d_0 +\ldots+ d_n,w_n
            }\qquad (n\in\omega) 
  \end{align*}
  \end{minipage}
  \caption{Structural rules for while-loops}
  \label{fig:rule:while}
  \end{figure*}
\begin{lemma}\label{lem:eq-inf-op}
For every closed program $p$, $p\sssto{d} \none\,$ implies $p\To d$.
\end{lemma}

\begin{proof}
By Lemma~\ref{lem:finite_chains}~(2), we must consider two cases.

(a) $p\sssto{e} p'\ssto{r} \none\,$ with $d = e+r$. It follows by induction
over the derivation of $p'\ssto{r} \none$ that $p'\To r$. As in the proof of
Lemma~\ref{lem:eq-finite-op}, by subsequent induction over the length of the chain
$p\ssto{e_1}\ldots\ssto{e_n} p'$, we obtain $p\To d$.

(b) $p\ssto{d_1} p_1\ssto{d_2} p_2\ldots,$ with $d=\sum\nolimits_i d_i$. We proceed
by induction over the number of computation constructs in $p$. Unless, $p$ is
a sequential composition or a while-loop, we can apply the induction hypothesis to
$p_1\ssto{d_2} p_2\ldots$ and thus obtain $p\ssto{d_1} p_1$ with $p_1\To\sum\nolimits_{i>1} d_i$
from which the goal follows as in clause~(a). Consider the two remaining cases.

Let $p$ be $\mbind{x\ass p'}{q}$. Then, either $p'\ssto{d_1} p_1'\ldots \ssto{d_n}\now w$,
hence
\begin{displaymath}
  p\ssto{d_1} \mbind{x\ass p'_1}{q}\ldots \ssto{d_n}\mbind{x\ass \now w}{q}\ssto{0} q[w/x]
\end{displaymath}
and we are done by applying the induction hypothesis to the derivation
$q[w/x]\ssto{d_{n+1}}\ldots$, or there is an infinite chain
$p'\ssto{d_1} p_1'\ssto{d_2}\ldots$, from which we obtain by induction
that $p'\To\sum\nolimits_i d_i$ and therefore $p\To\sum\nolimits_i d_i$, by
\textbf{(seq$_{\bm1}^{\mathbf{d}}$)}.

Finally, consider the case $p=(\mwhile{x\ass p'}{b}{q})$. If there is
an infinite chain $p'\ssto{d_1} p_1'\ssto{d_2}\ldots$ the proof
reduces to the induction hypothesis analogously to the case of
sequential composition above, but now using the
rule~\textbf{(wh$_{\bm1}^{\mathbf{d}}$)}. Otherwise,
$p'\sssto{\sum\nolimits_{i\leq n_1} {d_i}} \now{w_1}$ for suitable
$n_1$ and $w_1$. 
We then iterate this case distinction, which results in a sequence of
the form
\begin{align*}
  \mwhile{x\ass p'}{b}{q}      &\sssto{\sum\nolimits_{i\leq n_1} {d_i}}
  \mwhile{x\ass q[w_1/x]}{b}{q}\\ &\sssto{\sum\nolimits_{n_1<i\leq n_2} {d_i}}
  \mwhile{x\ass q[w_2/x]}{b}{q} ~\ldots
\end{align*}
which either terminates if for some $i$, $q[w_i/x]$ produces an
infinite chain, yielding a proof by reduction to the induction
hypothesis, or does not terminate. In the latter case, by
Lemma~\ref{lem:eq-finite-op},
$p'\To \sum\nolimits_{i\leq n_1} {d_i}, w_1$, and
$q[w_k/x]\To \sum\nolimits_{n_k<i\leq n_{k+1}} {d_i}$ for every
$k\in\omega$. Therefore, $p\To \sum\nolimits_i d_i$ by
\textbf{(wh$_{\bm5}^{\mathbf{d}}$)}.
\end{proof}

\noindent
Since for every closed computation term $p$ we can build a
possibly non-terminating reduction chain $p\ssto{d_1} p_1\ssto{d_2}$ (by
applying the rules in Fig.~\ref{fig:small-dur-oper}), we immediately
obtain
\begin{corollary}[Totality]
For any closed computation term $p$, precisely one of the judgements $p\To d$ or
$p\To d,v$ is derivable with the rules in Figure~\ref{fig:dur-oper}.
\end{corollary}

\begin{figure*}[t]
%
\begin{minipage}{0.85\textwidth}%
  \begin{flalign*}
    \anonrule{(unit)}{}{\sem{\Gamma\vctx\star\colon 1}= \lambda\bar x.\,\star}
    &&
    \anonrule{(var)}{
	  }{
		  \sem{\Gamma \vctx x_i \colon A} = \lambda\bar x.\, x_i
	  }
    &&
    \anonrule{(sig)}{
		  \sem{\Gamma\vctx v\colon A} = h\qquad \vaco{f}\colon A\to B\in\Sigma
	  }{
		  \sem{\Gamma \vctx f(v) \colon B} = f\comp h
	  }
  \end{flalign*}\\[-2.5ex]
  \begin{flalign*}
    \anonrule{(true)}{}{\sem{\Gamma\vctx\True\colon 2} = \lambda\bar x.\,\True}
&&
    \anonrule{(false)}{}{\sem{\Gamma\vctx\False\colon 2} = \lambda\bar x.\,\False}
&&
    \anonrule{(prod)}{
		  \sem{\Gamma\vctx v\colon A} = h\qquad \sem{\Gamma\vctx w\colon B} = l
	  }{
		  \sem{\Gamma \vctx \brks{v,w} \colon A\times B} = \brks{h,l}
	  }\\[-2.5ex]
  \end{flalign*}
%

\dotfill 
\vspace{-2ex}
\begin{flalign*}
    \anonrule{(pm)}{
      \sem{\Gamma \vctx v\colon A\times B} = h\qquad
      \sem{\Gamma,x\colon A,y\colon B \cctx p \colon C} =l
    }{
      \sem{\Gamma \cctx \pcase{\brks{x,y} \ass v}{p} \colon C} = l\comp\brks{\id,h}
    }
&&
    \anonrule{(seq)}{
      \sem{\Gamma \cctx p\colon A}  = h\qquad
      \sem{\Gamma, x\colon A \cctx q \colon B} =l
    }{
      \sem{\Gamma \cctx \mbind{x\ass p}{q} \colon B} = l^\klstar
      \comp \tau \comp \brks{\id, h}
    }
\end{flalign*}\\[-2.5ex]
%
\begin{flalign*}
    \anonrule{(now)}{
      \sem{\Gamma\vctx v\colon A} = h
    }{
      \sem{\Gamma \cctx \now v \colon A} = \eta_A\comp h
    }
&&
    \anonrule{(if)}{
      \sem{\Gamma \vctx v\colon 2}  = b\qquad
      \sem{\Gamma \cctx p \colon A} = h\qquad \sem{\Gamma \cctx q \colon A} = l
    }{
      \sem{\Gamma \cctx \mif{v}{p}{q} \colon A} = \lambda\bar x.\,\ite{h(\bar x)}{b(\bar x)}{l(\bar x)}
    }
\end{flalign*}\\[-1ex]
%
\begin{align*}
    \anonrule{(tr)}{
        \sem{\Gamma,t\colon\real\vctx v\colon A}  = h~\qquad
        \sem{\Gamma,x\colon A\vctx w\colon 2}       = b~\qquad 
        \lambda\bar x.\, \sup\{e\mid \forall t\in [0,e].\, b(\bar x,h(\bar x,t)) \} =d
    }{
      \sem{\Gamma \cctx \cwhile{x\ass t.\,v}{w} \colon A} = \lambda\bar x.\, \ite{\brks{d(\bar x),h(\bar x,d(\bar x))}}{b(\bar x,h(\bar x,d(\bar x)))}{d(\bar x)}
    } 
\end{align*}\\[-1.5ex]
  \begin{align*}
    \anonrule{(wh)}{
		  \sem{\Gamma\cctx p\colon A} = h \qquad \sem{\Gamma,x\colon A\vctx v\colon 2} = b\qquad \sem{\Gamma,x\colon A\cctx q\colon A} = l
	  }{
		  \sem{\Gamma\cctx \mwhile{x\ass p}{v}{q} \colon A} = \lambda\bar x.\, ((\lambda x.\,\ite{(Q\inr)\comp l(\bar x,x)}{b(\bar x,x)}{\eta\inl x})^\istar)^\klstar (h(\bar x))
	  }
  \end{align*}
\end{minipage}    
\caption{Denotational duration semantics}\label{fig:Q-denote}
\end{figure*}
\section{Denotational Duration Semantics}\label{sec:Q-sem}
Next we build a syntax-independent, \emph{denotational} counterpart to
the operational semantics previously introduced. In order to achieve
this, we will take advantage of the general principles of fine-grain
call-by-value~\cite{LevyPowerEtAl02,GeronLevy16,GoncharovRauchEtAl18},
and thus introduce a monad for the duration semantics.  Recall that
$(-)^\klstar$ denotes the Kleisli lifting of a monad $\BBT$, the set
$\real_{\mplus}$ denotes the non-negative real numbers and the set
$\bar\real_{\mplus}$ denotes the non-negative real numbers with
infinity.
\begin{definition}[Duration monad]\label{defn:Q-set}
The duration monad $\BBQ$ is defined by the following data:
$QX = \real_{\mplus}\times X\cup \bar\real_{\mplus}$, $\eta(x) = \brks{0,x}$, and
\begin{align*}
  (f\colon X\to QY)^\klstar(d,x) &= \brks{d+e,y} &&\text{\qquad if~~} f(x) = \brks{e,y},\\
  (f\colon X\to QY)^\klstar(d,x) &= d + e && \text{\qquad if~~} f(x) = e,\\
  (f\colon X\to QY)^\klstar(d)   &= d.
\end{align*}
\end{definition}
\noindent
It is straightforward to check that $\BBQ$ is a monad using the fact
that~$\realp$ is a monoid and that~$\realpe$ is a monoid
$\realp$-module.
\begin{theorem}
$\BBQ$ is a monad.
\end{theorem}
\begin{proof}
  In every \emph{distributive category} $\BC$~\cite{Cockett93},
  specifically $\BC=\Set$, for every monoid $M\in |\BC|$ and every
  monoid $M$-module $E\in |\BC|$ (i.e.\ an object with a monoid action
  $M\times E\to E$ satisfying obvious equations),
  $M\times\argument + E$ canonically extends to a (strong)
  monad~\cite{CoumansJacobs13}. We regard this as folklore. With $E$
  being the initial object,
  $M\times\argument + E\cong M\times\argument$ is known as the
  \emph{writer monad} or as the \emph{(monoid) action monad}.  In that
  sense, the duration monad can be identified as a \emph{generalised}
  writer monad with~$\realp$ viewed as a monoid under addition
  and~$\realpe$ (which is not initial) viewed as a monoid module under
  addition $\realp\times\realpe\to\realpe$ of possibly extended real
  numbers.
\end{proof}
Now, to be able to interpret while-loops over
$\BBQ$, we need to equip it with an \emph{iteration operator}
$(-)^\istar$,
\begin{align*}
  \infer{f^\istar\colon X\to QY}{f\colon X\to Q(Y+X)}
\end{align*}
Moreover, we expect $\BBQ$ to satisfy the classical laws of
iteration~\cite{BloomEsik93}, \ie\ to be (completely)~Elgot.
\begin{definition}[Elgot monads~\cite{Elgot75,AdamekMiliusEtAl11}]\label{def:elgot-monad}
  A strong mo\-nad~$\BBT$ on a distributive category is a
  \emph{(complete) Elgot monad} if it is equipped with an iteration
  operator sending each $f\colon X\to T(Y+X)$ to $f^\istar\colon X\to TY$ and
  satisfying the following principles:
\begin{itemize}
  \item\emph{fixpoint law}:  $f^\istar = [\eta,f^\istar]^\klstar  f$;
  \item\emph{naturality}: $g^{\klstar} f^{\istar} = ([(T\inl) \comp g, \eta\inr]^{\klstar} f)^{\istar}$ for $f\colon X\to T(Y+X)$, $g \colon Y \to TZ$;
  \item\emph{codiagonal}: $(T[\id,\inr] \comp f)^{\istar} = f^{\istar\istar}$ for  $f \colon X \to T((Y + X) + X)$;
  \item\emph{uniformity}: $f \comp h = T(\id+ h) \comp g$ implies
	$f^{\istar} \comp h = g^{\istar}$ for $f\colon X \to T(Y + X)$, $g\colon Z \to T(Y + Z)$ and
	$h\colon Z \to X$;
  \item\emph{strength}: $\tau\comp(\id_Z\times f^\istar) =((T\dist)\comp\tau\comp (\id_Z\times f))^\istar$ for any $f\colon X\to T(Y+X)$.
\end{itemize}
\end{definition}

\noindent
A standard strategy to define an Elgot monad structure on $\BBQ$ would be to
enrich (the Kleisli category~of) $\BBQ$ over complete partial orders,
and then compute $f^\istar$ as the least fixpoint of the continuous
map $[\eta,\argument]^\klstar f\colon\Hom(X, QY)\to \Hom(X, QY)$. However,
this would not compatible with our duration semantics, because the iteration
operator of $\BBQ$ must capture \emph{non-inductive} behaviours of
while-loops postulated in~\textbf{(wh$_{\bm5}^{\mathbf{d}}$)}, the latter
implying that even \emph{divergent} computations produce meaningful
durations. %
%
We thus proceed to make $\BBQ$ into an \emph{Elgot monad} in a
different way.  To that end we will first introduce a fine-grained
version of $\BBQ$, roughly mimicking the small-step duration
semantics.
\begin{definition}[Layered duration monad, Guardedness]\label{def:Q-layer}
Let $\hat QX=\nu\gamma.\,(X+\realp\times\gamma)$, which extends to 
a monad $\hat\BBQ$ by general results~\cite{Uustalu03}. On $\Set$ this monad
takes a concrete particularly simple form: $\hat QX = \realp^\star\times X\cup \realp^\omega$,
$\eta(x)= \brks{\eps,x}\in \hat QX$, and
\begin{align*}
  (f\colon X\to \hat QY)^\klstar(w,x) &= \brks{wu,y} &&\text{\qquad if~~} f(x) = \brks{u,y}\in\realp^\star\times Y,\\
  (f\colon X\to \hat QY)^\klstar(w,x) &= wu &&\text{\qquad if~~} f(x) = u\in \realp^\omega,\\
  (f\colon X\to \hat QY)^\klstar(w) &= w.\\[-5ex]
\end{align*}
A morphism $f\colon X\to \hat Q(Y+Z)$ is called \emph{guarded} when for
every $x \in X$ if $f(x) = (w, \inr z)$ then the word $w$ is
non-empty.  Guardedness of $f\colon X\to \hat Q(Y+X)$ intuitively means that
the iteration operator applied to $f$ successively produces durations
as the loop unfolds, in other words that the loop is
progressive. This provides a fine-grained semantics of divergence, for
every infinite stream of durations thus obtained contributes as a
specific divergent behaviour.
\end{definition}

\begin{proposition}\label{prop:layer-Q-fix}
  For a guarded morphism $f\colon X\to \hat Q(Y+X)$, there is a unique
  ${f^\istar\colon X\to \hat QY}$, satisfying the fixpoint equation
  $f^\istar = [\eta,f^\istar]^\klstar f$.
\end{proposition}
\begin{proof}
  In the coalgebraic form, this was shown in previous work for a large
  class of monads~\cite{Uustalu02,Milius05} including $\hat\BBQ$.
\end{proof}
We now render $\BBQ$ as a quotient of $\hat\BBQ$ and derive a (total)
Elgot iteration operator on~$\BBQ$ from the partial one on $\hat\BBQ$,
as follows.
Note that every $QX$ is a quotient of $\hat QX$ under the ``weak
bisimulation'' relation $\approx$ generated by the clauses,
\begin{flalign*}\label{eq:Q-bisim}
&&\brks{r_1\ldots r_n,x} \approx\,& \brks{s_1\ldots s_m,x}&& (r_1+\ldots+r_n = s_1+\ldots+s_m)\\
&&r_1r_2\ldots  \approx\,& s_1s_2\ldots  && \Bigl(\sum\nolimits_i r_i = \sum\nolimits_j s_j\Bigr)
\end{flalign*}
Then let $\rho_X\colon\hat QX\to QX$ be the emerging quotienting map:
\begin{flalign*}
  \rho_X(r_1\ldots r_n,x) = \Brks{\sum\nolimits_i r_i,x}, \qquad
  \rho_X(r_1r_2\ldots) = \sum\nolimits_i r_i,
\end{flalign*}
and let $\upsilon_X:QX\to\hat QX$ be defined in the following way:
\begin{flalign*}
\upsilon_X(d,x) = \brks{d,x}, \quad
\upsilon_X(r\in QX\cap\realp) = r0^\omega,\quad
\upsilon_X(\infty) = 1^\omega.
\end{flalign*}
Note that by definition, for any $f\colon X\to Q(Y+X)$, $\upsilon f\colon X\to\hat Q(Y+X)$
is guarded.
\begin{theorem}\label{thm:Q-elgot}
The following is true for $\rho_X$ and $\upsilon_X$: 
\begin{enumerate}
  \item every $\rho_X$ is a left inverse of the corresponding $\upsilon_X$;
  \item $\rho$ is a monad morphism;
  \item for any guarded $f\colon X\to\hat Q(Y+X)$, $\rho f^\istar = \rho(\upsilon\rho f)^\istar$.
\end{enumerate}
\end{theorem}
\begin{proof}
  (1) is easy to see by definition, and~(2) is essentially due to the
  fact that countable summation of non-negative real numbers is
  associative. Let us check~(3).
For every $x_0\in X$, there are three possibilities:
\begin{enumerate}[label=(\alph*)]
  \item there is a finite sequence of elements $x_0,\ldots,x_n$ in $X$ such that
$f(x_i) = \brks{w_i,x_{i+1}}$ for $i< n$, $f(x_n) = \brks{w_n,y}\in\realp^\mplus\times Y$
and $f^\istar(x_0) = \brks{w_0\ldots w_n,y}$;
  \item there is a finite sequence of elements $x_0,\ldots,x_n$ in $X$ such that
$f(x_i) = \brks{w_i,x_{i+1}}$ for $i< n$, $f(x_n) = w_n\in\realp^\omega$
and $f^\istar(x_0) = w_0 \ldots w_n$;
  \item there is an infinite sequence of elements $x_0,\ldots$ in $X$ such that
$f(x_i) = \brks{w_i,x_{i+1}}\in\realp^\mplus\times X$ for $i\in\omega$ and $f^\istar(x_0) = w_0 w_1\ldots$
\end{enumerate}
These three cases fully describe how the iteration operator of
$\hat \BBQ$ works. In each case, it is easy to check the identity in
question: e.g.\ in~(a)
$\rho f^\istar (x_0) = \rho \brks{w_0\ldots w_n\comma y} = \Brks{\sum
  w_0+\ldots +\sum w_n\comma y}$ and then
$\rho (\upsilon\rho f)^\istar (x_0)=\rho\Brks{\sum w_0\ldots\sum
  w_n\comma y} = \Brks{\sum w_0+\ldots +\sum w_n\comma y}$ where by
$\sum w_i$ we mean the sum of elements of the vector
$w_i\in\realp^\mplus$. The remaining cases~(b) and~(c) are handled
analogously.
\end{proof}

\noindent We immediately obtain the following corollary of Theorem~\ref{thm:Q-elgot}.
\begin{corollary}\label{cor:Q-elgot}
$\BBQ$ is an Elgot monad with the Elgot iteration sending
$f\colon X\to Q(Y+X)$ to $\rho\comp (\upsilon f)^\istar\colon X\to QY$.
\end{corollary}
 \begin{proof}
  In the terminology of~\cite{GoncharovSchroderEtAl17}, Theorem~\ref{thm:Q-elgot}
 implies that the pair $(\rho,\upsilon)$ is an \emph{iteration-congruent retraction}.
 Thus, by~\cite[Theorem 21]{GoncharovSchroderEtAl17}, $\BBQ$ is indeed an Elgot monad
 under the iteration operator described above.
\end{proof}

\noindent
Note that the Elgot iteration of $\BBQ$ \emph{adds} the durations
produced by a map $f \colon X \to Q(Y + X)$ along the respective iteration
process, the latter realised by feeding back to $f$ values of type $X$
it produces. For example, if $f(x) = \brks{1, \inr x}$ then
$\rho(\upsilon f)^\istar(x) = 1 + 1 + \dots = \infty$; if
$f(n+1) = \brks{1/2, \inr n}$ and $f(0) = \brks{0, \inl \star}$ then
$\rho(\upsilon f)^\istar(n) = \brks{n * (1/2),\star}$.

Using Elgot iteration provided by Corollary~\ref{cor:Q-elgot}, we can
now define a denotational semantics of~\mbox{\hc} in the following
way.  We identify the types in~\eqref{eq:sum-types} with what they
denote in $\Set$, i.e.\ natural numbers, real numbers, the one element
set $\{\star\}$, the two-element set $\{\top,\bot\}$ and Cartesian
products respectively. For a variable context
$\Gamma=(x_1\colon A_1,\ldots,x_n\colon A_n)$, let
$\sem{\Gamma}=A_1\times\ldots\times A_n$ (specifically,
$\sem{\Gamma}=1$ if $\Gamma$ is empty). We will use~$\bar x$ to refer
to the tuples $\brks{x_1,\ldots, x_n}$ where
$\Gamma=(x_1\colon A_1,\ldots,x_n\colon A_n)$.
%
The semantics
\mbox{$\sem{\Gamma\vctx v\colon A}$} and $\sem{\Gamma\cctx p\colon A}$ are respectively
maps of type $\sem{\Gamma}\to A$ and $\sem{\Gamma}\to QA$, inductively
defined in Figure~\ref{fig:Q-denote}.

\begin{theorem}[Soundness and Adequacy]
  \label{thm:sound_adeq}
Given a computation judgement $\argument\cctx p\colon A$,
\begin{align*}
 p\To d   &\text{\qquad iff\qquad} \sem{\argument\cctx p\colon A} = d, \\
 p\To d,v &\text{\qquad iff\qquad} \sem{\argument\cctx p\colon A} = \brks{d,\sem{\argument\vctx v\colon A}}.
\end{align*}
\end{theorem}
\begin{proof}
First we establish \emph{soundness}, which means simultaneously that
$p\To d$ implies $\sem{\argument\cctx p\colon A} = d$ and $p\To d,v$
implies $\sem{\argument\cctx p\colon A} = \brks{d,\sem{\argument\vctx v\colon A}}$.  By
induction over the derivation length in the proof system of
Figure~\ref{fig:dur-oper} this amounts to checking that every rule
preserves these implications. This is done routinely for most of the
rules. Consider some instances.

\textbf{(tr$_{\bm1}^{\mathbf{d}}$)} The assumption reads as $\cwhile{x\ass t.\,v}{w} \To d, v[d/t]$.
The premises of the rule imply
\begin{align*}
 d\geq&\;\sup\{e\mid \forall s\in [0,e].\, w[v[s/t]/x] = \True\},\\
 d\leq&\;\sup\{e\mid \forall s\in [0,e].\, w[v[s/t]/x] = \True\},
\end{align*}
hence $d=\sup\{e\mid \forall s\in [0,e].\, w[v[s/t]/x] = \True\}$. Assuming
$\sem{t\colon\real\vctx v\colon A} = h$ and $\sem{x\colon A\vctx w\colon 2} = b$, the latter can be rewritten
as
\begin{displaymath}
  d=\sup\{e\mid \forall t\in [0,e].\, b(h(t))\}.
\end{displaymath}
The left premise of~\textbf{(tr$_{\bm1}^{\mathbf{d}}$)}
also implies that $b(h(d))=\True$, and hence the corresponding rule in Figure~\ref{fig:Q-denote}
produces $\sem{\argument\cctx\cwhile{x\ass t.\,v}{w}\colon A} = \brks{d,h(d)}$ as required.

\textbf{(tr$_{\bm2}^{\mathbf{d}}$)} The assumption reads as $\cwhile{x\ass t.\,v}{w} \To d$.
Assuming $\sem{\argument\vctx t\colon\real\vctx v\colon A} = h$ and $\sem{\argument\vctx x\colon A\vctx w\colon 2} = b$, 
like in the previous clause, the premises of the rule imply $d=\sup\{e\mid \forall t\in [0,e].\, b(h(t))\}$,
however, now $b(h(d))=\False$. By using the same rule from Figure~\ref{fig:Q-denote},
we obtain $\sem{\argument\cctx\cwhile{x\ass t.\,v}{w}\colon A} = d$.

\textbf{(wh$_{\bm4}^{\mathbf{d}}$)} Suppose that $\mwhile{x\ass p}{v}{q} \To d + r, u$ is
obtained from the respective premises of~\textbf{(wh$_{\bm4}^{\mathbf{d}}$)} and
$\sem{\argument\cctx p\colon A} = h$, $\sem{x\colon A\vctx v\colon 2} = b$, $\sem{x\colon A\cctx q\colon A} = l$.
We obtain that $h=\brks{d,w}$
for some $w$ (by induction), $b(w)=\True$ and $\sem{\argument\cctx\mwhile{x\ass q[w/x]}{v}{q}\colon A} = \brks{r,\sem{\argument\vctx u\colon A}}$ (by induction).
Let
\begin{displaymath}
  f=\lambda x.\,\ite{(Q\inr)(l(x))}{b(x)}{\eta(\inl x)}
\end{displaymath}
and note that $\brks{r,\sem{\argument\vctx u\colon A}} = (f^\istar)^\klstar (l(w))$, which means that
$l(w) = \brks{e,c}$, $f^\istar(c) = \brks{e',\sem{\argument\vctx u\colon A}}$ and
${e+e' = r}$ for suitable $e$, $e'$ and $c$. Therefore, using the 
fixpoint identity for $(\argument)^\istar$,
$\sem{\argument\cctx\mwhile{x\ass p}{v}{q}\colon A}= (f^\istar)^\klstar (h)= (f^\istar)^\klstar (d,w)
= [\eta,f^\istar]^\klstar f^\klstar (d,w) = (f^\istar)^\klstar (d+e\comma c) = \brks{d+e+e'\comma \sem{\argument\vctx u\colon A}}
= \brks{d+r,\sem{\argument\vctx u\colon A}}$
and we are done.

\textbf{(wh$_{\bm5}^{\mathbf{d}}$)} Suppose that $\mwhile{x\ass p}{v}{q} \To \sum\nolimits_i d_i$
is obtained from the premises 
$p\To d_0,w_0$,
$q[w_i/x] \To d_{i+1}\comma w_{i+1}$ $(i\in\omega)$, and 
$v[w_i/x] = \True$ $(i\in\omega)$ by~\textbf{(wh$_{\bm5}^{\mathbf{d}}$)}.
By induction, $\sem{\argument\cctx p\colon A} = \brks{d_0,\sem{\argument\vctx w_0\colon A}}$, 
and for every $i\in\omega$, $\sem{\argument\cctx q[w_i/x]\colon A} = l(\sem{\argument\vctx w_{i}\colon A}) 
= \brks{d_{i+1},\sem{\argument\vctx w_{i+1}\colon A}}$ where $l=\sem{x\colon A\cctx q\colon A}$.
Note also that $b(\sem{\argument\vctx w_i\colon A}) = \True$ for every $i\in\omega$ with
$b=\sem{x\colon A\vctx v\colon 2}$.
Let again
\begin{displaymath}
  f=\lambda x.\,\ite{(Q\inr)(l(x))}{b(x)}{\eta(\inl x)}
\end{displaymath}
and observe that $f(\sem{\argument\vctx w_{i}\colon A}) = (Q\inr)(l(\sem{\argument\vctx w_{i}\colon A})) = \brks{d_{i+1},\inr\sem{\argument\vctx w_{i+1}\colon A}}$. 
Now, by definition of $(\argument)^\istar$
of $\hat\BBQ$, $(\upsilon f)^\istar\newline(\sem{\argument\vctx w_{0}\colon A}) = d_1d_2\ldots$, and therefore in
$\BBQ$, 
$f^\istar(\sem{\argument\vctx w_{0}\colon A}) 
= \rho(\upsilon f)^\istar(\sem{\argument\vctx w_{0}\colon A}) 
= \sum\nolimits_{i>0} d_i$. 
Therefore,
$\sem{\argument\cctx\mwhile{x\ass p}{v\\ }{q}\colon A} = 
(f^\istar)^\klstar(d_0,\sem{\argument\vctx w_{0}\colon A})= d_0 + f^\istar(\sem{\argument\vctx w_{0}\colon A})=\sum\nolimits_i d_i$,
as desired.

The remaining adequacy direction reads as follows: 
$\sem{\argument\cctx p\colon A} = d$ implies $p\To d$
and $\sem{\argument\cctx p\colon A} = \brks{d,\sem{\argument\vctx v\colon A}}$
implies $p\To d,v$. In order to obtain it by induction over
derivations in Figure~\ref{fig:Q-denote}, we slightly strengthen the
induction invariant as follows: for every substitution $\sigma$
sending variables from $\Gamma$ to values of corresponding types,
$\sem{\Gamma\cctx p\colon A} = d$ implies $p\sigma\To d\sigma$ and
$\sem{\Gamma\cctx p\colon A} = \brks{d,\sem{\Gamma\vctx v\colon A}}$ implies
$p\sigma\To d\sigma,v\sigma$.  It is then verified routinely that
every rule in Figure~\ref{fig:Q-denote} preserves this invariant.
\end{proof}

\begin{figure*}
\begin{minipage}{.98\textwidth}
  \begin{flalign*}
    \lrule{(seq$_{\bm1}^{\mathbf{e}}$)}{
        (p,s\TTo v_s)_{s\leq t}
      \qquad
        (q[v_s/x],0\TTo w_s)_{s\leq t}
    }{
      \mbind{x\ass p}{q},t \TTo w_t
    }
    &&
    \lrule{(seq$_{\bm2}^{\mathbf{e}}$)}{
      p \To d, v' \!\qquad (p,s\TTo v_s)_{s\leq d} \!\qquad 
        (q[v_s/x],0\TTo w_s)_{s\leq d} \!\qquad q[v_d/x],t\TTo w 
    }{
      \mbind{x\ass p}{q},d+t \TTo w
    }
  \end{flalign*}\\[-3ex]
%
%
\begin{flalign*} 
\lrule{(if$_{\bm1}^{\mathbf{e}}$)}{
  p,t\TTo v
}{
  \mif{\True}{p}{q},t \TTo v
}
&& 
\lrule{(if$_{\bm2}^{\mathbf{e}}$)}{
  q,t\TTo v
}{
  \mif{\False}{p}{q},t \TTo v
}
&& 
\lrule{(pm$^{\mathbf{e}}$)}{
  p[v/x,w/y], t \TTo u
}{
  \pcase{\brks{x,y}\ass\brks{v,w}}{p},t \TTo u
}
\end{flalign*}\\[-3ex]
\begin{flalign*}
    \lrule{(now$^{\mathbf{e}}$)}{
      }{
        \now v,0 \TTo v
      }
&&
   \lrule{(tr$_{\bm1}^{\mathbf{e}}$)}{
		\cwhile{x \ass s.\,v}{w} \To d\qquad t<d
	}{
		\cwhile{x\ass s.\,v}{w},t \TTo v[t/s]
	}&&
  \lrule{(tr$_{\bm2}^{\mathbf{e}}$)}{
		\cwhile{x \ass s.\,v}{w} \To d, u \qquad t\leq d
	}{
		\cwhile{x \ass s.\,v}{w},t \TTo v[t/s]
	}
\end{flalign*}\\[-3ex]
  %
%
\begin{flalign*}
    \lrule{(wh$_{\bm1}^{\mathbf{e}}$)}{p \To d \qquad p,t \TTo w \qquad {t < d}}{
		  \mwhile{x\ass p}{v}{q},t \TTo w
	  }
&&
  \lrule{(wh$_{\bm3}^{\mathbf{e}}$)}{
      p \To d,w' \qquad p,d \TTo w \qquad v[w/x] = \False}{
		  \mwhile{x\ass p}{v}{q},d  \TTo w
  }
  \end{flalign*}
%
\begin{flalign*}
    \lrule{(wh$_{\bm2}^{\mathbf{e}}$)}{
      p \To d, w' \qquad p,t \TTo w \qquad {t < d}
	  }{
		  \mwhile{x\ass p}{v}{q},t \TTo w
	  }
  &&
    \lrule{(wh$_{\bm4}^{\mathbf{e}}$)}{
      p \To d,w' \qquad p \TTo d,w\qquad v[w/x] = \True \qquad \mwhile{x\ass q[w/x]}{v}{q},t\TTo u}{
		  \mwhile{x\ass p}{v}{q},d+t \TTo u
	  }
\end{flalign*}\\[-3ex]
\end{minipage}
  \caption{Big-step evolution operational semantics}\label{fig:hyb-oper}
\end{figure*}
\section{Evolution Semantics}\label{sec:evo}

Building on the duration semantics of the previous section, we will
now present a big-step \emph{evolution operational semantics} for
hybrid programs, which incorporates the former as a critical
ingredient. In particular, both semantics will be needed for providing
the aforementioned adequacy result. The new semantics relates a closed
computation term $p$ with its \emph{trajectory}, i.e.\ the range of
time-indexed values produced by the execution of $p$ -- recall for
instance the bouncing ball and cruise controller described in
Section~\ref{sec:hybcore}. More specifically, the corresponding
big-step relation $\TTo$ connects a closed computation term $p$ and a
time instant~$t$ with the value~$v$ to which $p$ evaluates at~$t$, in
symbols $p,t\TTo v$. The corresponding derivation rules are presented
in Figure~\ref{fig:hyb-oper}. These rules recur to the duration
semantic judgements governed by the rules in
Figure~\ref{fig:dur-oper}.  Unlike duration semantics, evolution
semantics is not total: for example, for a given computation $p$,
$p,t\TTo v$ may not be derivable for any $v$ and $t$, which yields an
\emph{empty trajectory} for $p$, denotationally a totally undefined
function $\emptyset \to A$.

Let us discuss the non-obvious rules of this semantics, starting
with~\textbf{(seq$_{\bm1}^{\mathbf{e}}$)} which takes infinite
families of judgements as premises. This rule is similar to
\textbf{(seq$_{\bm1}^{\mathbf{d}}$)}, but
in~\textbf{(seq$_{\bm1}^{\mathbf{e}}$)} even if the first computation
evaluates to $v_t$ at $t$, we cannot disregard~$q$ and use~$v_t$ in
the conclusion, because the return type of $q$ need not match the
return type of $p$.  The evaluation $q[v_t/x],0\TTo w_t$ is required
to perform the necessary type conversion. Moreover, the rule premise
ensures that for every $s\leq t$, $q[v_s/x]$ also yields a value at
$0$. To see why this is necessary, consider the following instance of
$\mbind{x\ass p}{q}$:
\begin{displaymath}
  \mbind{x\ass (\cwhile{x\ass t.\,t}{\True})}{\mif{x< 1}{(\cwhile{x\ass 1}{\False})}{x}}
\end{displaymath}
First observe that for any given time instant $t \geq 0$ the
subprogram $\cwhile{x\ass t.\,t}{\True}$ evaluates to $t$, intuitively
modelling the passage of time \emph{ad infinitum} and thus producing a
trajectory of infinite duration.  On the other hand, the subprogram
$\cwhile{x\ass 1}{\False}$ models non-progressive divergence and
\emph{does not} yield a value at any time instant, in particular
$q[0/x], 0 \TTo v$ for no $v$. Hence the entire program
$\mbind{x\ass p}{q}$ diverges already at~$0$, as we cannot apply any
of the sequential composition rules relative to this time
instant. However, $q[1/x], 0 \TTo 1$ and therefore if we did not check
previous time instants for divergence (using $q[v_s/x],0 \TTo w_s$ for
$s \leq 1$) we would obtain $\mbind{x\ass p}{q}, 1 \TTo 1$ by the
sequential composition rules, meaning that we would not be able to
detect divergence of the entire program.
\begin{remark}\label{rem:Q-H-diff}
  The evolution semantics subtly differs from the duration
  semantics. Specifically, the above example indicates that $p\To d$
  does not necessarily imply that
  $d=\sup\{ t\mid \exists w.\, p,t\TTo w\}$. 
\end{remark}
The main subtlety behind our semantics of while-loops
via~\textbf{(wh$_{\bm1}^{\mathbf{e}}$)}--\textbf{(wh$_{\bm4}^{\mathbf{e}}$)}
is that, contrary to standard program semantics, while-loops no longer
behave as iterated sequential composition, which would be easily
achievable~via the rule
 \begin{displaymath}
    \frac{
      \begin{array}[b]{ll} 
        &(p,s\TTo v_s)_{s\leq t} \qquad v[v_{t}/x]=\True\\[1.5ex]
        &(\mwhile{x\ass q[v_s/x]}{v}{q},0\TTo w_s)_{s\leq t,v[v_s/x]=\True}
      \end{array}
	  }{
		  \mwhile{x\ass p}{v}{q},t \TTo w_t
	  }\quad
 \end{displaymath}
instead of~\textbf{(wh$_{\bm1}^{\mathbf{e}}$)} and similarly
for~\textbf{(wh$_{\bm2}^{\mathbf{e}}$)}--\textbf{(wh$_{\bm4}^{\mathbf{e}}$)}. 
We argue, however, that our rules are more suitable for hybrid
programming: the ones corresponding to iterated sequential composition
produce undesirable artefacts, e.g.\ the program
\begin{flalign*}
  \mwhile{x \ass 0}{\True}{\cwhile{x \ass t .\, (t + x) }{1}}
\end{flalign*}
would not evaluate under these rules at any time instant (we would not
be able to apply any of these rules since at every time instant we would be
lacking a terminating condition) and the same would apply to the
bouncing ball in Example~\ref{expl:ball} and other systems involving
infinite while-loops or Zeno behaviour. Intuitively, the reader may
see the adopted rules as the ones corresponding to testing the
condition of the while-loop only at the end of the input trajectory
whilst assuming that the other points evaluate to false; the rules
corresponding to iterated sequential composition would, on the other
hand, correspond to testing the condition of the while-loop at
\emph{all} points of the input trajectory.

The following lemma is obtained by induction over the derivation of the 
corresponding computation judgement by the rules in Figure~\ref{fig:hyb-oper}.
\begin{figure*}
  \begin{minipage}{0.85\textwidth}%
    \centering
%
%
$\displaystyle
  \anonrule{(tr)}{
      \sem{\Gamma,t\colon\real\vctx v\colon A} = h \qquad \sem{\Gamma, x\colon A\vctx w\colon 2} = b\qquad
      \lambda \bar x.\,\sup \{e\in\realp \mid \forall t\in [0,e].\, b(\bar x,h(\bar x,t))\} = d}  
      {\sem{\Gamma \cctx \cwhile{x\ass t.\,v}{w} \colon A} = \lambda\bar x.\, \brks{\ite{\cc}{b(\bar x,h(\bar x, d(\bar x)))}{\od}, d(\bar x), \lambda t.\, h(\bar x,t)}
  }
$\\[3ex]
$\displaystyle
  \anonrule{(wh)}{
    \sem{\Gamma\cctx p\colon A} = h \qquad \sem{\Gamma,x\colon A\vctx v\colon 2} = b\qquad \sem{\Gamma,x\colon A\cctx q\colon A} = l
	}{
    \sem{\Gamma\cctx \mwhile{x\ass p}{v}{q} \colon A} =
    \lambda\bar x.\, ((\lambda\brks{x,c}.\,\ite{(H\inr) (\kappa\, l(\bar x,x))}{b(\bar x,x) \wedge c}{\eta(\inl x)})^\istar)^\klstar (\kappa(h(\bar x)))
	}
$
\end{minipage}
\caption{Denotational evolution semantics }\label{fig:H-denote}
\end{figure*}
\begin{lemma}\label{lem:d-close}
For every closed computation term $p$, the set 
\begin{displaymath}
   \{t\in\realp\mid\exists v.\, p,t\TTo v\}
\end{displaymath}
is downward closed.
\end{lemma}

To relate the operational semantics with its denotational counterpart,
recall the hybrid monad $\BBH$ from our previous
work~\cite{GoncharovJakobEtAl18a} (in a slightly reformulated
equivalent form).  Intuitively, for every set~$X$, the elements
of~$HX$ are trajectories in the sense discussed above and Kleisli
composition is essentially used for concatenating trajectories whilst
simultaneously tracking divergence along the computation.  In what
follows we use the superscript notation~$e^t$ explained in
Section~\ref{sec:prelim}.
\begin{definition}[Hybrid monad]
The \emph{hybrid monad} $\BBH$ is defined as follows.~Let
\begin{align*}
  H X =\; &\{\brks{\cc, d, e:[0,d]\to X}\mid d\in\realp\}\; \cup\\
          & \{\brks{\cd, d, e:[0,d]\to X}\mid d\in\realp\}\; \cup \\
          & \{\brks{\od,d, e:[0,d)\to X}\mid d\in\realpe\},
\end{align*}
i.e.\ the set $HX$ contains three types of trajectory: \emph{closed
  convergent} $\brks{\cc, d, e}$, \emph{closed divergent}
$\brks{\cd, d, e}$ and \emph{open divergent} $\brks{\od, d, e}$
(notably, there are no ``open convergent'' ones), e.g., $HX$ contains
the \emph{empty trajectory} $\brks{\od,0,\bang}$, which we denote by
$\div$.  Given $p=\brks{\alpha,d,e}\in HX$, let us write
$p_\tt = \alpha$, $p_\dr = d$, $p_\ev = e$, $\dom e=[0,d)$ if
$\alpha=\od$ and $\dom e=[0,d]$ otherwise. For a map $f \colon X \to H Y$
we will abbreviate the notation $(f(x))_\tt$ to $f_\tt(x)$, and
similarly for the other cases, namely $(f(x))_\dr$, $(f(x))_\ev$ and
$(f(x))^t$.
The monad structure of $\BBH$ is then given in the following way:
$\eta(x) = \brks{\cc,0,\lambda t.\,x}$; for $f\colon X\to HY$,
$\brks{\alpha,d,e}\in HX$, let
\begin{align*}
  D=\bigcup\,\bigl\{[0,t]\subseteq\dom e \mid \forall s\in [0,t].\, f(e^s)\neq\div\bigr\}, 
\end{align*}
  $d' = \sup D$. Then $f^\klstar(\alpha,d,e)$ is defined by the clauses:
%
%
\begin{flalign*}
f^\klstar (\cc,d,e) =&\; \brks{f_{\tt}(e^d),d+f_\dr(e^d),\\
&\erule~\erule~ \lambda t.\,\ite{f_\ev^0(e^t)}{t<d}{f_\ev^{t-d}(e^d)}}\hspace{-1cm} & (\text{$D=[0,d]$})\\
f^\klstar (\cd,d,e) =&\; \brks{\cd,d,\lambda t.\, f_\ev^0(e^t)}
& (\text{$D=[0,d]$}) \\
f^\klstar (\alpha,d,e) =&\; \brks{\cd,d',\lambda t.\, f_\ev^0(e^t)}
& (\text{$D=[0,d']$, $d'<d$})\\
f^\klstar (\alpha,d,e) =&\; \brks{\od,d',\lambda t.\, f_\ev^0(e^t)}
& (\text{$D=[0,d')$, $d'\leq d$})
\end{flalign*}
%
%
%
%
\end{definition}

\noindent
Intuitively, only the first clause for $f^\klstar(\alpha,d,e)$
describes a successful `concatenation' scenario, which assumes that
the argument trajectory is closed and additionally ensures that
$f(e^s)$ does not yield divergence along the run of
$\brks{\cc,d,e}$. The remaining clauses exhaustively cover divergence
scenarios, caused either by~$f(e^s)$ or by the argument
trajectory. 
As shown in \cite{GoncharovJakobEtAl18a}, $\BBH$ is an Elgot monad. We
thus reuse the rules from Figure~\ref{fig:Q-denote} to define the
semantics over~$\BBH$ by replacing $Q$ with $H$ throughout, except for
the rules for~\textbf{(tr)} and~\textbf{(wh)} which are overridden by
those in Figure~\ref{fig:H-denote}. To cope with the non-standard
behaviour of while-loops discussed above, we call on an auxiliary
natural transformation $\kappa\colon H \to H (\argument \times 2)$,
appending to every point of the trajectory a Boolean flag signalling
whether the current point is the endpoint, formally:
\begin{align*}
\kappa(\alpha, d,e) = \brks{\alpha, d, u}, &&
u^t =
\begin{cases}
\brks{e^t,\top}& \text{if $\alpha=\cc$, $t=d$},\\
\brks{e^t,\bot}& \text{otherwise}.
\end{cases}
\end{align*}
Finally, we obtain our main result.
\begin{theorem}[Soundness and Adequacy]\label{thm:adeq-H}
Given $\argument\cctx p\colon A$,
\begin{enumerate}
  \item $p, t \TTo v$ iff $\sem{\argument\cctx p\colon A}=\brks{\alpha,d,e}$, $t\leq d$ and $e^t = v$;
  \item $p, d \TTo v$ and $p\To d,w$ iff $v=w$, $\sem{\argument\cctx p\colon A}=\brks{\cc,d,e}$ and $e^d = v$.
\end{enumerate}
\end{theorem}
\noindent
%
The first clause ensures that the
correspondence $t\mto v_t$ generated by the judgements $p, t \TTo v_t$ 
yields a trajectory that agrees with the denotational semantics
previously presented. 
The second clause is needed
to 
relate the duration and evolution semantics: both must agree in the
evaluation of the \emph{last point} of the trajectory produced by a
program $p$, and naturally both must agree with the corresponding
denotational semantics.

The theorem follows from a sequence of lemmas, which we present
next. In the rest of the section, we subscript the semantic brackets
with $\BBQ$ and $\BBH$ to refer to duration semantics and to evolution
semantics correspondingly. The first lemma provides a technical characterization
of iterations occurring as the semantics of while-loops.
\begin{lemma}\label{lem:H-elgot}
Let $f\colon X\to H(Y+X)$ have the property
\begin{align}\label{eq:it-as-while}
\forall x\in X.\,(\exists y\in X.\,f_\ev^t(x)=\inr y)\impl f_\tt(x)=\cc\land f_\dr(x)=t.  
\end{align} 
We define $f_i\colon X\to H(Y+X)$ by induction as follows: $f_0 = \eta\inr$, $f_{i+1} = [\eta\inl, f_i]^\klstar f$.
Let $x\in X$, and suppose that $f^\istar(x)=\brks{\alpha,d,e}$.
\begin{enumerate}
  \item If $t < d$ then there exists $n \in \nat$ such that $t < (f_n)_\dr(x)$ and $(f_n)_\ev^t(x) = \inl e^t$.
  \item If $\alpha\neq\od$ then there exists $n \in \nat$ such that $d=(f_n)_\dr(x)$ and $(f_n)_\ev^{d}(x) = \inl e^{d}$.
\end{enumerate}
\end{lemma}

\begin{proof}
  We rely on the fact that $\BBH$ is characterized as an iteration-congruent
retract of a monad $\BBH_0\BBM$ obtained by composing $\BBH_0$, an
iteration-free version of $\BBH$, with the \emph{maybe-monad} $\BBM$ (with
$MX=X+1$)~\cite{GoncharovJakobEtAl18a}. That is, there is a monad morphism
$\rho\colon\BBH_0\BBM\to\BBH$, whose components $\rho_X$ come with right inverses
$\upsilon_X$ and $\rho$ preserves iteration as follows: $\rho f^\iistar =
(\rho f)^\istar$. Here, by $(\argument)^\iistar$ we refer to the iteration of
$\BBH_0\BBM$, which is suitably characterized as the least solution of the
fixpoint law.
For the sake of the present argument it suffices to know that $H_0MX$ is obtained
from $HX$ by allowing the trajectories to be partial functions and by imposing the 
the following constraints: 
\begin{itemize}
  \item $\brks{\cc,d,e}\in H_0MX$ implies that $e^d$ is defined;
  \item closed divergent trajectories are not in $H_0MX$.
\end{itemize}
The
natural transformation~$\upsilon:H\to H_0M$ sends $\brks{\cc,d,e}\in HX$ to 
$\brks{\cc,d,e}\in H_0MX$ and $\brks{\alpha,d,e}\in HX$ with $\alpha\in\{\cd,\od\}$ 
to $\brks{\od,\infty,e'}\in H_0MX$ where $e'$ is the extension of $e$ to $\realpe$
with undefined elements. Conversely, $\rho(\alpha,d,e)$ is $\brks{\alpha,d,e}$ for total $e$ and
$\rho(\alpha,d,e)=\brks{\alpha',d',e'}$ for partial~$e$, in which case $e'$ is
a restriction of $e$ to the largest downward closed subset of the domain of 
definiteness of~$e$ and~$d'$ and $\alpha\in\{\cd,\od\}$ are selected accordingly. 

Let $g = \upsilon
f\colon X\to H_0M(Y+X)$. For every $i$, let $g^{\brks{i}}\colon X\to H_0M Y$ be defined as
follows: $g^{\brks{0}}=\div$ (the empty trajectory) and $g^{\brks{i+1}} =
[\eta,g^{\brks{i}}]^\klstar g$. We use the following property of 
$\BBH_0\BBM$~\cite[Lemma~25]{GoncharovJakobEtAl18a}: 
\begin{flalign}\label{eq:H-prop1}
\forall x\in X.\,g^\iistar(x) = \brks{\alpha,d,e}\land e^t\in Y \impl \exists n.\,(g^{\brks{n}})_{\ev}^t(x) = \inl e^t
\end{flalign}
for every $t\leq d$ for which $e^t$ is defined.
Let us define $g_0 = \eta\inr$, $g_{i+1} = [\eta\inl, g_i]^\klstar g$ and note 
that, by induction, for every~$i$, $g^{\brks{i}} = [\eta, \div]^\klstar g_i$. Therefore
we can reformulate~\eqref{eq:H-prop1} as follows: 
\begin{flalign}\label{eq:H-prop2}
\forall x\in X.\,g^\iistar(x) = \brks{\alpha,d,e}\land e^t\in Y  \impl \exists n.\,(g_n)_{\ev}^t(x) = \inl e^t.
\end{flalign}
In order to obtain the first clause of the lemma, suppose that
$f^\istar(x)=\brks{\alpha,d,e}$ and $t < d$. Now
$f^\istar(x)=\rho g^\iistar(x) = \brks{\alpha,d,e}$ and therefore
$g^\iistar(x) = \brks{\alpha',d',e'}$ with suitable $\alpha'$, $d'$,
$e'$ and $(e')^t=e^t$ because~$e'$ is at least as defined as
$e$. By~\eqref{eq:H-prop2}, for some $n$,
$(g_n)_{\ev}^t(x) = \inl e^t$. Note that by definition $g$
satisfies~\eqref{eq:it-as-while} and for every $y\in X$ the trajectory $g_\ev(y)$ is
total. It follows by induction that the same is true for $g_n$, hence
$g_n(x) = \upsilon (f_n(x))$, in particular, $t < (f_n)_\dr(x)$ and
$(f_n)_{\ev}^t(x) = \inl e^t$.  The second clause is shown
analogously.
\end{proof}

\noindent
Next, we check that the duration semantics and the evolution semantics
agree in the following sense (cf.~Remark~\ref{rem:Q-H-diff}).
\begin{lemma}\label{lem:H-cc}
Given a computation judgement $\Gamma\cctx p\colon A$, let $\bar v\in\sem{\Gamma}$.
Suppose that $\sem{\Gamma\cctx p\colon A}_\BBH(\bar v) = \brks{\alpha, d,e}$ and 
$\sem{\Gamma\cctx p\colon A}_\BBQ(\bar v) = \brks{d',v}$. Then $\alpha = \cc$ 
iff $d=d'$ and in either case $e^d=v$.
\end{lemma}
\begin{proof}[Proof Sketch]
The proof runs by induction over the structure of $p$. We restrict to the 
following three representative cases.
\begin{citemize}
  \item $p=(\mbind{x\ass q}{r})$. Let $f = \sem{\Gamma\cctx q\colon B}_\BBH$,
$g = \sem{\Gamma,x\colon B\cctx r\colon A}_\BBH$, $\hat f = \sem{\Gamma\cctx q\colon B}_\BBQ$,
$\hat g = \sem{\Gamma,x\colon B\cctx r\colon A}_\BBQ$. Suppose that ${\alpha = \cc}$.
Then $\sem{\Gamma\cctx p\colon A}_\BBH(\bar v) = g^\klstar\tau\brks{\bar v,f(\bar v)} = 
\brks{\cc,d,e}$, which by definition means that 
\begin{align*}
f(\bar v) =&\; \brks{\cc,d_1,e_1},&
g(\bar v, e_1^{d_1}) =&\; \brks{\cc, d_2, e_2},&
d =&\; d_1 + d_2,&
\intertext{and, analogously}
\hat f(\bar v) =&\; \brks{d_1',v'},&
\hat g(\bar v,v') =&\; \brks{d_2', v},
& d'=&\;d_1'+d_2'.&
\end{align*}
By induction, $d_1=d_1'$ and $e_1^{d_1}=v'$. Hence, again by induction, $d_2=d_2'$
and $e_2^{d_2} = v$. This implies $d = d'$ and $e^d = e_2^{d_2} = v$.
By reversing this argument we also obtain the implication in the
right-to-left direction.
  \item $p=(\cwhile{x\ass t.\,v}{w})$. By definition, 
  \begin{align*}
    \sem{\Gamma\cctx p\colon A}_\BBH(\bar v) =&\; \ite{\brks{\cc, u(\bar v), \lambda t.\, h(\bar v,t)}}{b(\bar v, h(\bar v, u(\bar v)))}{\\ &\; \brks{\od, u(\bar v), \lambda t.\, h(\bar v,t)}},\\
    \sem{\Gamma\cctx p\colon A}_\BBQ(\bar v) =&\; \ite{\brks{u(\bar v),h(\bar v,u(\bar v))}}{b(\bar v,h(\bar v,u(\bar v)))}{u(\bar v)}
  \end{align*}
  where $h=\sem{\Gamma,t\colon\real\vctx v\colon A}$ and $b=\sem{\Gamma, x\colon A\vctx w\colon 2}$, $u=\lambda\bar x.\,\sup \{e\in\realp \mid \forall t\in [0,e].\, b(\bar x,h(\bar x,t))\}$. The assumption implies 
$b(\bar v, h(\bar v, u(\bar v)))=\top$, hence we reduce to
\begin{align*}
  \sem{\Gamma\cctx p\colon A}_\BBH(\bar v) =&\; \brks{\cc, u(\bar v), \lambda t.\, h(\bar v,t)},\\
  \sem{\Gamma\cctx p\colon A}_\BBQ(\bar v) =&\; \brks{u(\bar v),h(\bar v,u(\bar v))}
\end{align*}
and the claim of the lemma holds trivially.
  \item $p=(\mwhile{x\ass q}{v}{r})$. By definition, 
  \begin{align*}
    \sem{\Gamma\cctx p\colon A}_\BBH(\bar v) =&\; ((\lambda\brks{x,c}.\,\ite{(H\inr) (\kappa\, l(\bar v,x))\\*&\qquad}{b(\bar v,x) \wedge c}{\eta(\inl x)})^\istar)^\klstar (\kappa\, h(\bar v)),\\*
    \sem{\Gamma\cctx p\colon A}_\BBQ(\bar v) =&\; ((\lambda x.\,\ite{(Q\inr)(\hat l(\bar v,x))\\*&\qquad}{b(\bar v,x)}{\eta(\inl x)})^\istar)^\klstar (\hat h(\bar v))
  \end{align*}
where $b=\sem{\Gamma,x\colon A\vctx v\colon 2}$, and 
\begin{align*}
h=&\;\sem{\Gamma\cctx q\colon A}_\BBH,& l=&\;\sem{\Gamma,x\colon A\cctx r\colon A}_\BBH,&\\
\hat h=&\;\sem{\Gamma\cctx q\colon A}_\BBQ,& \hat l=&\;\sem{\Gamma,x\colon A\cctx r\colon A}_\BBQ.
\end{align*}
Suppose that $\sem{\Gamma\cctx q\colon A}_\BBH(\bar v) = \brks{\cc, d,e}$. Then, by definition of $\BBH$, $d = d_0 + d_{\star}$
where
\begin{align*}
\brks{\cc, d_0,e_0} &\,= h(\bar v),\\
\brks{\cc,d_{\star},e_{\star}} &\,= \\
(\lambda\brks{x,c}.\,&\ite{(H\inr) (\kappa\, l(\bar v,x))}{b(\bar v,x) \wedge c}{\eta\inl x})^\istar(e_0^{d_0},\top).
\end{align*}
Consider the sequence 
\begin{align}\label{eq:H-while}
  \brks{\cc, d_0,e_0},\brks{\cc, d_1,e_1},\ldots
\end{align}
formed as follows: 
$\brks{\cc, d_{i+1},e_{i+1}} = l(\bar v,e_i^{d_i})$ if $b(\bar v,e_i^{d_i})$. That is, every triple 
$\brks{\cc, d_{i+1},e_{i+1}}$ is obtained by applying the above expression under
dagger to $\brks{e_i^{d_i},\top}$ as long as $b(\bar v,e_i^{d_i})$ is true. 
Note that this process cannot produce triples with the first element different 
than $\cc$, because this would mean that the loop diverges after finitely many 
iterations, returning a triple different from $\brks{\cc,d_{\star},e_{\star}}$.
By Lemma~\ref{lem:H-elgot} (2), the sequence~\eqref{eq:H-while} is necessarily finite.
Hence, by induction hypothesis $\sem{\Gamma,x\colon A\cctx r\colon A}_\BBQ(\bar v, e_i^{d_i}) = \brks{d_{i+1}, e_{i+1}^{d_{i+1}}}$ 
for every $i$, from which we obtain $\sem{\Gamma\cctx p\colon A}_\BBQ(\bar v)=\brks{d,e^d}$,
as desired.

For the converse, suppose that $\sem{\Gamma\cctx p\colon A}_\BBQ(\bar v) = \brks{d,v}$. By definition
of iteration in $\BBQ$, this means that the loop in $\sem{\Gamma\cctx p\colon A}_\BBQ(\bar v)$
is unfolded finitely many times, i.e.\ $\brks{d_0,v_0}=\hat h(\bar v)$, 
$\brks{d_1,v_1} = \hat l(\bar v,v_0),\ldots,\brks{d_n,v_n} = \hat l(\bar v,v_{n-1})$
and $b(\bar v, v_0)=\top,\ldots$, $b(\bar v, v_{n-1})=\top$, ${b(\bar v, v_0)=\bot}$ for suitable 
$d_i$, $v_i$ and $d = d_0+\ldots+d_n$, $v=v_n$. By induction hypothesis, 
$\sem{\Gamma\cctx p\colon A}_\BBH(\bar v) = \brks{\cc,d_0,e_0}$,
$\sem{\Gamma,x\colon A\cctx q\colon A}_\BBH(\bar v,v_i) = \brks{\cc, d_{i+1},e_{i+1}}$ for $i=1,\ldots,n$ 
and $e_{i}^{d_i} = v_i$ for every $i=0,\ldots,n$. Using the above expression 
for $\sem{\Gamma\cctx p\colon A}_\BBQ(\bar v)$ and the assumption that 
$\sem{\Gamma\cctx p\colon A}_\BBH(\bar v) = \brks{\alpha, d,e}$, we obtain that 
$\alpha=\cc$. 
\qed
\end{citemize}
\noqed\end{proof}
\noindent
The following lemma essentially captures the soundness direction of
Theorem~\ref{thm:adeq-H}.
\begin{lemma}\label{lem:sound-H}
Given $-\cctx p\colon A$,
\begin{enumerate}
  \item if $p, t \TTo v$ then $\sem{-\cctx p\colon A}_\BBH=\brks{\alpha,d,e}$, $t\leq d$ and $e^t = v$;
  \item if $p, d \TTo v$ and $p\To d,w$ then $v=w$, $\sem{-\cctx p\colon A}_\BBH=\brks{\cc,d,e}$ and $e^d = v$.
\end{enumerate}
\end{lemma}
\begin{proof}
First, let us argue that for a fixed $-\cctx p\colon A$,~(1) 
implies~(2). Indeed, $p, d \TTo v$ by the first clause implies that 
$\sem{-\cctx p\colon A}_\BBH=\brks{\alpha,d,e}$, $e^d = v$ and by 
Theorem~\ref{thm:sound_adeq}, $\sem{-\cctx p\colon A}_\BBQ = \brks{d,w}$.
By Lemma~\ref{lem:H-cc}, $e^d = w = v$ and $\alpha=\cc$.

The proof of~(1) is now obtained analogously to the proof of Theorem~\ref{thm:sound_adeq},
by induction over the derivation $p, t \TTo v$ in the proof system of Figure~\ref{fig:hyb-oper}.
We occasionally need to call both~(1) and~(2) in the induction step, but, as we 
argued~(1) implies~(2) for every specific~$p$. Consider, e.g.\ the rule
\begin{align*}
\lrule{(seq$_{\bm2}^{\mathbf{e}}$)}{
  \begin{array}{c@{\qquad}c}
    p \To d, v' & (p,s\TTo v_s)_{s\leq d} \\[1ex] 
    (q[v_s/x],0\TTo w_s)_{s\leq d} & q[v_d/x],t\TTo w 
  \end{array}
}{
  \mbind{x\ass p}{q},d+t \TTo w
}
\end{align*}
in detail. By induction hypothesis, $\sem{-\cctx p\colon A}_\BBH=\brks{\cc,d,e}$, 
${e^d=v_d}$, $\sem{-\cctx q[v_d/x]\colon B}_\BBH=\brks{\alpha,d_\star,e_\star}$, 
$t\leq d_\star$, $e_\star^t=w$. Now, by definition, and using Lemma~\ref{lem:d-close},
\begin{align*}
\sem{-\cctx&\mbind{x\ass p}{q}\colon A}_\BBH\\ 
  =&\; \brks{f_{\tt}(e^d),d+f_\dr(e^d),\lambda s.\,\ite{f_\ev^0(e^s)}{s<d}{f_\ev^{s-d}(e^d)}}\\
  =&\; \brks{f_{\tt}(v_d),d+f_\dr(v_d),\lambda s.\,\ite{f_\ev^0(v_s)}{s<d}{f_\ev^{s-d}(v_d)}}
\end{align*}
where $f = \sem{x\colon B\cctx q\colon A}_\BBH$, and therefore $d+t\leq d+d_\star=  d+f_\dr(v_d)$
and 
\begin{displaymath}
  \ite{f_\ev^0(v_{d+t})}{{d+t<d}}{f_\ev^{d+t-d}(v_d)} = f_\ev^{t}(v_d) = e_\star^t = w,
\end{displaymath}
as desired.
\end{proof}
\noindent
The following lemma essentially captures the adequacy direction of Theorem~\ref{thm:adeq-H}.
\begin{lemma}\label{lem:adeq-H}
If $\sem{\argument\cctx p\colon A}_\BBH=\brks{\alpha,d,e}$ and $t\leq d$ then $p, t \TTo v$ for some $v$.
\end{lemma}
\begin{proof}[Proof Sketch]
  The proof is analogous to the corresponding fragment of the proof of
  Theorem~\ref{thm:sound_adeq}. Let us consider the case of
  while-loops, which is the trickiest one. That is, we have to
  show that
  $\sem{\argument\cctx \mwhile{x\ass p}{b}{q}\colon 
    A}_\BBH=\brks{\alpha,d,e}$ and $t \leq d$ imply that
  $\mwhile{x\ass p}{b}{q}, t \TTo v$ for some $v$. By expanding the
  assumption we obtain
\begin{align*}
 &\brks{\alpha,d,e}=\\
 &\bigl((\lambda\brks{x,c}.\,\ite{(H\inr) (\kappa(l(x)))}{b(x) \wedge c}{\eta(\inl x)})^\istar\bigr)^\klstar
 (\kappa(\alpha_0,d_0,e_0))
\end{align*}
where $b=\sem{x\colon A\vctx v\colon 2}$, $\brks{\alpha_0,d_0,v_0}=\sem{-\cctx p\colon A}_\BBH$ and $l=\sem{x\colon A\cctx q\colon A}_\BBH$. 
If $\alpha_0=\od$ or $\alpha_0=\cd$, the above equation yields $\alpha_0=\alpha$, $d=d_0$ and $e=e_0$
and we obtain the requisite judgement using the induction hypothesis either 
by \textbf{(wh$_{\bm1}^{\mathbf{e}}$)} or by \textbf{(wh$_{\bm2}^{\mathbf{e}}$)}.
The same considerations apply if $t<d_0$. We proceed under the assumption that $\alpha_0=\cc$ and $t\geq d_0$. Consider the sequence 
\begin{align}\label{eq:H-while-a}
  \brks{\cc, d_0,e_0},\brks{\cc, d_1,e_1},\ldots
\end{align}
iteratively constructed as follows: $d_{i+1}$ and $e_{i+1}$ are defined as long as 
$\brks{\cc,d_{i+1},e_{i+1}} = l(e_i^{d_i})$ and $b(e_i^{d_i})$ is true. If 
any of the latter conditions eventually fails, the sequence~\eqref{eq:H-while-a}
terminates resulting in $d = d_0+\ldots+d_n$ for some $t$. We then construct the 
requisite derivation of $\mwhile{x\ass p}{b}{q}, t \TTo v$ by calling the induction 
hypothesis. If the sequence~\eqref{eq:H-while-a} is infinite, still, by 
Lemma~\ref{lem:H-elgot} (1), for some $n$, $d_0+\ldots+d_n\leq t<d_0+\ldots+d_{n+1}$
and $e^t = e_n^{t - d_0-\ldots-d_n}$. Again, by induction, we can construct a 
derivation of the judgement $\mwhile{x\ass p}{b}{q}, t \TTo e_{n}^{t - d_0-\ldots-d_n}$.
\end{proof}
\noindent
Finally, we can prove the soundness and adequacy theorem that was
previously mentioned.
\begin{proof}[Proof (Theorem~\ref{thm:adeq-H})]
The Lemmas~\ref{lem:sound-H},~\ref{lem:adeq-H} and~\ref{lem:H-cc} jointly
imply the theorem as follows. The left-to-right directions for both
clauses are already in Lemma~\ref{lem:sound-H}. For the right-to-left direction of
the first clause, Lemma~\ref{lem:adeq-H} produces some value $v$, which, again
by the soundness direction must be the one for which $e^t=v$. For the right-to-left
direction of the second clause we additionally use Lemma~\ref{lem:H-cc}.
\rnnote{I'm not sure you can apply Lemma~\ref{lem:H-cc}, because you cannot
  assume the condition $\sem{\Gamma\cctx p\colon A}_\BBQ(\bar v) = \brks{d',v}$. I think one needs to consider two clauses in Lemma~\ref{lem:H-cc}}
\sgnote{Not should be OK?}
\end{proof}

\section{Conclusions and Further Work}\label{sec:conc}
Our present work is\rnnote{For future work: Could the small-step
  semantics generated from the big-step semantics in Sec. 4 be
  described combined inductively/coinductively instead of using
  infinite premises?}
a 
result of an ongoing effort to establish solid semantic foundations
for hybrid computation. More specifically, here we complement our
previous work on denotational
models~\cite{NevesBarbosaEtAl16,GoncharovJakobEtAl18,
  GoncharovJakobEtAl18a} in terms of \emph{hybrid monads}, with a
corresponding operational semantics and connect both styles of
semantics by a soundness and adequacy theorem. The central ingredient
of our framework is a hybrid, call-by-value while-language \hc{},
which we put forward as syntactic means for capturing the essence of
hybrid computation, subsequently study it, and use as basis of future,
more complex hybrid programming languages.  The task of formulating an
adequate operational semantics turned out to require a sophisticated
route via an auxiliary lightweight duration semantics, and we regard
the fact of success, i.e.\ the fact that the term ``adequate
operational semantics'' can actually be sensibly interpreted in the
hybrid setting, as a striking outcome of this work.

%
%
\hc{} provides a minimal framework combining the basic and
uncontroversial features of hybrid computation with clear and
principled semantic foundations. We plan to use it both as a stepping
stone and as a yardstick in further work.
The fact that~\hc{} purposefully combines the classical view of
(instantaneous) computations -- which either succeed immediately or
silently diverge (e.g.\ variable assignments in a loop) -- and
processes extended over real-time, makes the operational semantics
of~\hc{} particularly distinct, specifically, the rules generally
require \emph{continuum-size} numbers of premises, which is needed to
ensure the integrity of the resulting trajectories. This kind of
complexity vanishes when restricting to progressive semantics, which
is obtained by forbidding empty trajectories.  The
submonad~$\BBH_\mplus$ of $\BBH$ that arises from this restriction,
was studied in~\cite{GoncharovJakobEtAl18} and enjoys better properties
than~$\BBH$; we expect our adequacy result to restrict
accordingly. This task is however not entirely trivial, as it requires
a considerable adaptation of the language to keep track of
progressiveness of terms, yet a systematic approach to this issue
was proposed recently~\cite{GoncharovRauchEtAl18}.

As indicated in the introduction, the use of monads gives access to
various generic scenarios for combining hybridness with other effects,
specifically via universal constructions and monad transformers. One
generic construction that can be applied to $\BBH$ is the
\emph{generalized coalgebraic resumption monad
  transformer}~\cite{GoncharovSchroderEtAl18} sending an endofunctor
$T$ to the coalgebra $T_F=\nu\gamma.\ T(\argument+F\gamma)$ for a
given endofunctor $F$, which can be chosen suitably for modelling
\emph{interactive/concurrent} features, e.g.\ $F = A\times\argument$
corresponds to automata or process algebra-style actions ranging over
$A$. By the abstract results in~\cite{GoncharovSchroderEtAl18} for an
Elgot monad~$\BBT$, $T_F$ canonically extends to an Elgot monad
$\BBT_F$, which is indeed the case for $\BBT=\BBH$. A disciplined way
of adding further effects to \hc{} is by resorting to Plotkin and
Power's \emph{generic effects}~\cite{PlotkinPower01}. E.g.\ for binary
non-determinism, we would need a coin-tossing primitive
$\Gamma\cctx\oname{toss}()\colon 2$ for choosing a control branch
non-deterministically, yielding a non-deterministic choice operator in
the following way (a similar approach applies e.g.\ to
probabilistic~choice):
\begin{displaymath}
  p + q = (\mbind{b\ass\oname{toss}()}{\mif{b}{p}{q}}).
\end{displaymath}
Denotationally, in order to interpret such extensions of \hc{} we need
to modify $\BBH$ by applying suitable monad transformers, in
particular, we plan to extend previous work on combining
iteration-free versions of $\BBH$ with other
monads~\cite{DahlqvistNeves18}.  Our commitment to the language of
categorical constructions raises hopes that the presented developments
can be carried over to categories other than~$\Set$ and more suitable
for analysing topological aspects of hybrid dynamics, such as
\emph{stability} of hybrid behaviour under input perturbations. E.g.\
the construction of the duration monad as a quotient of the
corresponding layered duration monad can be instantly reproduced in
the category $\Top$ of topological spaces, therefore inducing a
necessary topologization in a canonical way. It remains an important
goal of further work to provide an analogous treatment for the hybrid
monad~$\BBH$, possibly by introducing a \emph{layered hybrid monad}
whose quotient would be~$\BBH$.
%
%
Our long-term goal is to extend
\hc{} to a higher-order hybrid language with general recursion.


%
\begin{acks}
The first author would like to acknowledge the support of German Research Foundation
under grant~GO~2161/1\dash 2.

The second author would like to acknowledge the support of ERDF -- European
Regional Development Fund through the Operational Programme for Competitiveness
and Internationalisation - COMPETE 2020 Programme and by National Funds through
the Portuguese funding agency, FCT - Funda\c{c}\~ao para a Ci\^encia e a
Tecnologia, within project POCI-01-0145-FEDER-030947.
\end{acks}

%
\bibliographystyle{ACM-Reference-Format}
\bibliography{monads}

\end{document}

%% file: catprog.tex
\providecommand{\catname}{\mathbf} 
\providecommand{\clsname}{\mathcal}
\providecommand{\oname}[1]{\operatorname{\mathsf{#1}}}

\def\defcatname#1{\expandafter\def\csname B#1\endcsname{\catname{#1}}}
\def\defcatnames#1{\ifx#1\defcatnames\else\defcatname#1\expandafter\defcatnames\fi}
\defcatnames ABCDEFGHIJKLMNOPQRSTUVWXYZ\defcatnames

\def\defclsname#1{\expandafter\def\csname C#1\endcsname{\clsname{#1}}}
\def\defclsnames#1{\ifx#1\defclsnames\else\defclsname#1\expandafter\defclsnames\fi}
\defclsnames ABCDEFGHIJKLMNOPQRSTUVWXYZ\defclsnames

\def\defbbname#1{\expandafter\def\csname BB#1\endcsname{\mathbb{#1}}}
\def\defbbnames#1{\ifx#1\defbbnames\else\defbbname#1\expandafter\defbbnames\fi}
\defbbnames ABCDEFGHIJKLMNOPQRSTUVWXYZ\defbbnames

\def\Set{\catname{Set}}

\def\Top{\catname{Top}}

\providecommand{\eps}{\operatorname\epsilon}

\providecommand{\argument}{\operatorname{-\!-}}

\DeclareOldFontCommand{\bf}{\normalfont\bfseries}{\mathbf}
\providecommand{\mplus}{{\scriptscriptstyle\bf+}} 	

\providecommand{\Hom}{\mathsf{Hom}}
\providecommand{\id}{\mathsf{id}}

\providecommand{\comp}{\mathbin{\circ}}

\providecommand{\bang}{\operatorname!}				

\providecommand{\dar}{\kern-1.2pt\operatorname{\downarrow}}	
\providecommand{\uar}{\kern-1.2pt\operatorname{\uparrow}}	
\providecommand{\mto}{\mapsto}

\providecommand{\True}{\top}
\providecommand{\False}{\bot}

\providecommand{\impl}{\Rightarrow}

\providecommand{\fst}{\oname{fst}}
\providecommand{\snd}{\oname{snd}}

\providecommand{\brks}[1]{\langle #1\rangle}
\providecommand{\Brks}[1]{\bigl\langle #1\bigr\rangle}

\providecommand{\inl}{\oname{inl}}
\providecommand{\inr}{\oname{inr}}

\DeclareSymbolFont{Symbols}{OMS}{cmsy}{m}{n}
\DeclareMathSymbol{\iobj}{\mathord}{Symbols}{"3B}

\providecommand{\ev}{\oname{ev}}

\usepackage{stmaryrd}

\providecommand{\lsem}{\llbracket}
\providecommand{\rsem}{\rrbracket}
\providecommand{\sem}[1]{\lsem #1 \rsem}

\providecommand{\comma}{,\operatorname{}\linebreak[1]}		
\providecommand{\dash}{\nobreakdash-\hspace{0pt}}			    
\providecommand{\erule}{\rule{0pt}{0pt}}		        	    

\providecommand{\pacman}[1]{}					

\providecommand{\noqed}{\def\qed{}}				

\providecommand{\mone}{{\text{\kern.5pt\rmfamily-}\sf\kern-.5pt1}}

\makeatletter
\@ifpackageloaded{enumitem}{}{\usepackage[loadonly]{enumitem}}		
\makeatother

\newlist{citemize}{itemize}{1}
\setlist[citemize]{label=\labelitemi,wide} 

\newlist{cenumerate}{enumerate}{1}
\setlist[cenumerate,1]{label=\arabic*.~,ref={\arabic*},wide} 

\makeatletter
\def\mfix#1{\oname{#1}\@ifnextchar\bgroup\@mfix{}}	
\def\@mfix#1{#1\@ifnextchar\bgroup\mfix{}}			
\makeatother

\providecommand{\case}[3]{\mfix{case}{\mathbin{}#1}{of}{#2}{\kern-1pt;}{\mathbin{}#3}}